\documentclass[11pt,pdftex,a4paper]{article}%
\pdfoutput=1
\usepackage{float}
\floatstyle{boxed} 
\restylefloat{figure}
\usepackage{subfig}

\usepackage{verbatim}
\usepackage{amssymb}
\usepackage{lipsum}
\usepackage{multicol}
\usepackage{amsmath}
\usepackage{framed}
\usepackage{amsthm}
\usepackage{enumitem}
\setlist{nolistsep}
\usepackage{tikz}
\usepackage[margin=0.92in]{geometry}
\usepackage{color}
\usepackage{cite}
\usepackage{tabularx}
\newif\ifcode
\codefalse
\codetrue
\usepackage{macros}
\ifcode
\usepackage[ruled]{algorithm}
\usepackage{algpseudocode}
\usepackage{ioa_code}

\newtheorem{theorem}{Theorem}

\newtheorem{claim}[theorem]{Claim}

\newtheorem{definition}{Definition}
\newtheorem{lemma}[theorem]{Lemma}
\newtheorem{observation}[theorem]{Observation}

\newcounter{linenumber}

\newcommand{\true}{\mathit{true}}
\newcommand{\false}{\mathit{false}}

\newcommand{\remove}[1]{}

\newcommand{\Wset}{\textit{Wset}}
\newcommand{\Rset}{\textit{Rset}}
\newcommand{\Dset}{\textit{Dset}}

\newcommand{\txns}{\textit{txns}}

\newcommand{\Read}{\textit{read}}
\newcommand{\Write}{\textit{write}}
\newcommand{\TryC}{\textit{tryC}}

\newcommand{\ok}{\textit{ok}}

\newcommand{\ignore}[1]{}

\begin{document}
\bibliographystyle{abbrv}

\title{Why Transactional Memory Should Not Be Obstruction-Free}

\author{
Petr Kuznetsov$^1$~~~Srivatsan Ravi$^2$ \\
$^1$\normalsize T\'el\'ecom ParisTech \\
$^2$\normalsize TU Berlin
}

\date{}
\maketitle
\thispagestyle{empty}
\begin{abstract}
Transactional memory (TM) is an inherently optimistic abstraction:
it allows concurrent processes to execute sequences of
shared-data accesses (transactions) speculatively, with an option of
aborting them in the future. 
Early TM designs avoided using locks and 
relied on non-blocking  
synchronization to ensure \emph{obstruction-freedom}: 
a transaction that encounters no step contention 
is not allowed to abort. 
However, it was later observed that obstruction-free TMs perform
poorly and, as a result, state-of-the-art TM implementations are
nowadays \emph{blocking},
allowing aborts because of data \emph{conflicts} rather than step contention.

In this paper, we explain this shift in the TM practice theoretically, via complexity bounds.
We prove a few important lower bounds on obstruction-free TMs. Then we
present a \emph{lock-based} TM implementation that beats all of these lower bounds.
In sum, our results exhibit a considerable complexity gap
between non-blocking and blocking TM implementations.
\end{abstract}

\section{Introduction}
\label{sec:intro}
Transactional memory (TM) allows concurrent processes to
organize sequences of operations on shared \emph{data items} into atomic
transactions. 
A transaction may commit, in which case its updates of data items
``take effect'' or it may \emph{abort},  in which case no
data items are updated.
A TM \emph{implementation} provides processes with algorithms for implementing
transactional operations on data items (such as \emph{read}, \emph{write} and \emph{tryCommit})
by applying \emph{primitives} on shared \emph{base objects}.
Intuitively, the idea behind the TM abstraction is optimism: 
before a transaction commits, all its operations are
\emph{speculative}, and it is expected that, in the absence of
concurrency, a transaction commits. 
     

It therefore appears natural that early TMs
implementations~\cite{HLM+03, astm, ST95,nztm,fraser} 
adopted optimistic concurrency control and guaranteed
that a prematurely halted transaction cannot not prevent
other transactions from committing. 
These implementations avoided using locks and relied on \emph{non-blocking}
(sometimes also called \emph{lock-free})  synchronization. 
Possibly the weakest non-blocking
progress condition is \emph{obstruction-freedom}~\cite{HLM03,HS11-progress}
stipulating that  every transaction running in the absence
of \emph{step contention}, \emph{i.e.}, not encountering steps of concurrent
transactions, must commit.   

In 2005, Ennals~\cite{Ennals05} argued that 
that obstruction-free TMs inherently yield poor performance, because
they require  transactions to forcefully abort each other. 
Ennals further describes a \emph{lock-based} TM
implementation~\cite{Ennals-code}
that he claimed to outperform \emph{DSTM}~\cite{HLM+03},
the most referenced obstruction-free TM implementation at the time.
Inspired by \cite{Ennals05}, more recent TM implementations like \emph{TL}~\cite{DStransaction06}, 
\emph{TL2}~\cite{DSS06} and \emph{NOrec}~\cite{norec}
employ locking and showed that Ennal's claims about performance of
lock-based TMs hold true on most workloads. 
The progress guarantee provided by these TMs is typically \emph{progressiveness}: 
a transaction may be aborted only if it encounters a read-write or a
write-write conflicts with a concurrent transaction~\cite{GK09-progressiveness}.
 
There is a considerable amount of empirical evidence on the
performance gap between non-blocking
(obstruction-free) and blocking (progressive) TM implementations but, to
the best of our knowledge, 
no analytical result explains it.
Complexity lower and upper bounds presented  in this paper provide
such an explanation.
 
\vspace{1mm}\noindent\textbf{Lower bounds for non-blocking TMs.}
Our first result focuses on two important TM properties: \emph{weak
disjoint-access-parallelism} (weak DAP) and \emph{read invisibility}.
Weak DAP~\cite{AHM09} is believed to improve
TM performance by ensuring that transactions
\emph{concurrently contend} on the same base object (both access the base object and at least one updates it) only if their data 
sets are connected in the \emph{conflict graph} constructed on the
data sets of concurrent transactions~\cite{AHM09}.
Many popular obstruction-free TM implementations satisfy weak DAP~\cite{nztm,fraser,HLM+03}, but
not the stronger property of \emph{strict DAP}~\cite{ofgk,tm-book} that
disallows any two transactions to contend on a base object unless they access a common data item.

A TM implementation uses invisible reads if, informally, a reading transaction cannot
cause a concurrent transaction to abort (we give a more precise definition later in this paper), 
which is believed to be important for
(most commonly observed) read-dominated workloads. 
Interestingly, lock-based TM
implementations like \emph{TL}~\cite{DStransaction06}
are weak DAP and use invisible reads.
In contrast, we establish that it is impossible to
implement a \emph{strictly serializable} (all committed transactions appear to execute sequentially in some
total-order respecting the timing of non-overlapping transactions) obstruction-free TM that provides
both weak DAP and read invisibility. Indeed, obstructions TMs like \emph{DSTM}~\cite{HLM+03}
and \emph{FSTM}~\cite{fraser} satisfy weak DAP, but not read invisibility since read operations must write to the shared memory.

We then derive lower bounds on obstruction-free TM implementations with respect to
the number of \emph{stalls}~\cite{G05} 
The stall complexity captures the fact that 
the time a process might have to spend before it applies a primitive on a base object
can be proportional to the number of processes that try to concurrently update the object~\cite{G05}.  
Our second result shows that a single read operation in a $n$-process strictly serializable obstruction-free TM
implementation may incur $\Omega(n)$ stalls.

Finally, we prove that any \emph{read-write (RW) DAP} \emph{opaque} 
(all transactions appear to execute sequentially in some
total-order respecting the timing of non-overlapping transactions) obstruction-free 
TM implementation has an execution in which a
read-only transaction incurs $\Omega(n)$ non-overlapping \emph{RAWs} or \emph{AWARs}. 
Intuitively, RAW (read-after-write) or AWAR (atomic-write-after-read)
patterns~\cite{AGK11-popl} capture the amount of ``expensive
synchronization'', \emph{i.e.}, the number of costly conditional primitives
or memory barriers~\cite{AdveG96} incurred by the
implementation.
The metric appears to be more practically relevant than simple
step complexity, as it accounts for expensive cache-coherence operations or conditional instructions.
RW DAP, satisfied by most
obstruction-free implementations~\cite{HLM+03,fraser}, 
requires that read-only transactions 
do not contend on the same base object with transactions having disjoint write sets.
It is stronger than \emph{weak DAP}~\cite{AHM09}, but weaker than \emph{strict DAP}~\cite{OFTM}.  

\begin{figure}[h]
      
      \scalebox{1}[1]{
      \begin{tabularx}{\textwidth}{c|c|c}
	~~~~~ & Obstruction-free TMs & Our progressive TM $LP$ \\ \hline
	strict DAP & No~\cite{OFTM} & Yes \\ \hline
	invisible reads+weak DAP & No & Yes\\ \hline
	stall complexity of reads & $\Omega(n)$ 
        & $O(1)$  \\ \hline
	RAW/AWAR complexity & $\Omega(n)$ & $O(1)$  \\ \hline
	read-write primitives, wait-free termination & No~\cite{tm-book} & Yes
   \end{tabularx}
    
}
\caption{Complexity gap between blocking and non-blocking strictly serializable TM
  implementations; $n$ is the number of processes}\label{fig:main}
\end{figure}

\vspace{1mm}\noindent\textbf{An upper bound for blocking TMs.}
To exhibit a complexity gap between blocking and non-blocking TMs, we
describe a progressive opaque TM implementation that beats the impossibility result
and the lower bounds we established for obstruction-free TMs.     

Our implementation, denoted $LP$, (1)~uses only read and write primitives on base objects and
ensures that every transactional operation terminates in a wait-free
manner, (2)~ensures strict DAP, (3)~has invisible reads, 
(4)~performs $O(1)$ non-overlapping RAWs/AWARs per transaction, and
(5)~incurs $O(1)$ memory stalls
for read operations.
In contrast, the following claims hold for any implementation in the
class of obstruction-free (OF) strict serializable TMs:
No OF TM can be implemented (i)~using only read and
write primitives and provide wait-free termination~\cite{tm-book}, or (ii) provide strict DAP~\cite{OFTM}.
Furthermore, (iii)~no weak DAP OF TM has invisible reads
(Theorem~\ref{th:ir}) and (iv)~no OF TM ensures a constant
number of stalls incurred by a read operation (Theorem~\ref{th:oftmstalls}). 
Finally, (v)~no RW DAP opaque OFTM has constant  RAW/AWAR complexity
(Theorem~\ref{th:oftriv}).
In fact, (iv) and (v) exhibit a linear separation between blocking and non-blocking TMs w.r.t 
expensive synchronization and memory stall complexity, respectively.

Our results are summarized in Figure~\ref{fig:main}.
Altogether, we grasp a considerable complexity gap between blocking
and non-blocking TM implementations, justifying theoretically the
shift in TM practice we observed during the past decade. 

\vspace{2mm}\noindent\textbf{Roadmap.}
Sections~\ref{sec:model} and \ref{sec:prel} define our model and the classes of TMs 
considered in this paper.
Section~\ref{sec:oftm} contains lower bounds for  obstruction-free
TMs. 
Section~\ref{sec:iclb} describes our lock-based TM implementation $LP$.
In Section~\ref{sec:related}, we discuss the related
work and in Section~\ref{sec:disc},  concluding remarks.
Some proofs are delegated to the optional appendix. 
%

%
%
\section{Model}
\label{sec:model}
\vspace{1mm}\noindent\textbf{TM interface.}
\emph{Transactional memory} (in short, \emph{TM})
allows a set of data items (called \emph{t-objects}) to be accessed 
via atomic \emph{transactions}.
Every transaction $T_k$ has a unique identifier $k$. 
We make no assumptions on the \emph{size} of a t-object, \emph{i.e.}, 
the cardinality on the set $V$ of possible values a t-object can store.
A transaction $T_k$ may contain the following \emph{t-operations},
each being a matching pair of an \emph{invocation} and a \emph{response}:
$\Read_k(X)$ returns a value in $V$ 
or a special value $A_k\notin V$ (\emph{abort});
$\Write_k(X,v)$, for a value $v \in V$,
returns \textit{ok} or $A_k$;
$\TryC_k$ returns $C_k\notin V$ (\emph{commit}) or $A_k$.

\vspace{1mm}\noindent\textbf{TM implementations.}
We consider an asynchronous shared-memory system in which
a set of $n$ processes, communicate by applying \emph{primitives} on shared \emph{base objects}.
We assume that processes issue transactions sequentially i.e. a process starts a new transaction
only after the previous transaction has committed or aborted.
A TM \emph{implementation} provides processes with algorithms
for implementing $\Read_k$, $\Write_k$ and $\TryC_k()$
of a transaction $T_k$ by \emph{applying} \emph{primitives} from a set of shared \emph{base objects}, each of which is 
assigned an \emph{initial value}.
We assume that these primitives are \emph{deterministic}.
A primitive is a generic \emph{read-modify-write} (\emph{RMW}) procedure applied to a base object~\cite{G05,Her91}.
It is characterized by a pair of functions $\langle g,h \rangle$:
given the current state of the base object, $g$ is an \emph{update function} that
computes its state after the primitive is applied, while $h$ 
is a \emph{response function} that specifies the outcome of the primitive returned to the process.
A RMW primitive is \emph{trivial} if it never changes the value of the base object to which it is applied.
Otherwise, it is \emph{nontrivial}.

\vspace{1mm}\noindent\textbf{Executions and configurations.}
An \emph{event} of a transaction $T_k$ (sometimes we say \emph{step} of $T_k$)
is an invocation or response of a t-operation performed by $T_k$ or a 
RMW primitive $\langle g,h \rangle$ applied by $T_k$ to a base object $b$
along with its response $r$ (we call it a \emph{RMW event} and write $(b, \langle g,h\rangle, r,k)$).

A \emph{configuration} (of a TM implementation) specifies the value of each base object and 
the state of each process.
The \emph{initial configuration} is the configuration in which all 
base objects have their initial values and all processes are in their initial states.

An \emph{execution fragment} is a (finite or infinite) sequence of events.
An \emph{execution} of a TM implementation $M$ is an execution
fragment where, starting from the initial configuration, each event is
issued according to $M$ and each response of a RMW event $(b, \langle
g,h\rangle, r,k)$ matches the state of $b$ resulting from all
preceding events.
An execution $E\cdot E'$ denotes the concatenation of $E$ and execution fragment $E'$,
and we say that $E'$ is an \emph{extension} of $E$ or $E'$ \emph{extends} $E$.

Let $E$ be an execution fragment.
For every transaction (resp., process) identifier $k$,
$E|k$ denotes the subsequence of $E$ restricted to events of
transaction $T_k$ (resp., process $p_k$).
If $E|k$ is non-empty,
we say that $T_k$ (resp., $p_k$) \emph{participates} in $E$, else we say $E$ is \emph{$T_k$-free} (resp., \emph{$p_k$-free}).
Two executions $E$ and $E'$ are \emph{indistinguishable} to a set $\mathcal{T}$ of transactions, if
for each transaction $T_k \in \mathcal{T}$, $E|k=E'|k$.
A TM \emph{history} is the subsequence of an execution consisting of the invocation and 
response events of t-operations.
Two histories $H$ and $H'$ are \emph{equivalent} if $\txns(H) = \txns(H')$
and for every transaction $T_k \in \txns(H)$, $H|k=H'|k$.

The \emph{read set} (resp., the \emph{write set}) of a transaction $T_k$ in an execution $E$,
denoted $\Rset(T_k)$ (resp., $\Wset(T_k)$), is the set of t-objects that $T_k$ reads (resp., writes to) in $E$.
More specifically, if $E$ contains an invocation of $\Read_k(X)$ (resp., $\Write_k(X,v)$), we say
that $X\in \Rset(T_k)$ (resp., $\Wset(T_k)$).
The \emph{data set} of $T_k$ is $\Dset(T_k)=\Rset(T_k)\cup\Wset(T_k)$.
A transaction is called \emph{read-only} if $\Wset(T_k)=\emptyset$; \emph{write-only} if $\Rset(T_k)=\emptyset$ and
\emph{updating} if $\Wset(T_k)\neq\emptyset$.
Note that we consider the conventional dynamic TM programming model: 
the data set of a transaction is not known apriori (\emph{i.e.}, at the start of the transaction)
and it is identifiable only by the set of t-objects the transaction has invoked a read or write in the given execution.

\vspace{1mm}\noindent\textbf{Transaction orders.}
Let $\txns(E)$ denote the set of transactions that participate in $E$.
An execution $E$ is \emph{sequential} if every invocation of
a t-operation is either the last event in the history $H$ exported by $E$ or
is immediately followed by a matching response.
We assume that executions are \emph{well-formed} i.e. for all $T_k$, $E|k$ begins with the invocation of a t-operation, is
sequential and has no events after $A_k$ or $C_k$.
A transaction $T_k\in \txns(E)$ is \emph{complete in $E$} if
$E|k$ ends with a response event.
The execution $E$ is \emph{complete} if all transactions in $\txns(E)$
are complete in $E$.
A transaction $T_k\in \txns(E)$ is \emph{t-complete} if $E|k$
ends with $A_k$ or $C_k$; otherwise, $T_k$ is \emph{t-incomplete}.
$T_k$ is \emph{committed} (resp., \emph{aborted}) in $E$
if the last event of $T_k$ is $C_k$ (resp., $A_k$).
The execution $E$ is \emph{t-complete} if all transactions in
$\txns(E)$ are t-complete.

For transactions $\{T_k,T_m\} \in \txns(E)$, we say that $T_k$ \emph{precedes}
$T_m$ in the \emph{real-time order} of $E$, denoted $T_k\prec_E^{RT} T_m$,
if $T_k$ is t-complete in $E$ and
the last event of $T_k$ precedes the first event of $T_m$ in $E$.
If neither $T_k\prec_E^{RT} T_m$ nor $T_m\prec_E^{RT} T_k$,
then $T_k$ and $T_m$ are \emph{concurrent} in $E$.
An execution $E$ is \emph{t-sequential} if there are no concurrent
transactions in $E$.
We say that $\Read_k(X)$ is \emph{legal} in a t-sequential execution $E$ if it returns the
\emph{latest written value} of $X$ in $E$, and $E$ is \emph{legal}
if every $\Read_k(X)$ in $E$ that does not return $A_k$ is legal in $E$.

\vspace{1mm}\noindent\textbf{Contention.}
We say that a configuration $C$ after an execution $E$ is \emph{quiescent} (resp., \emph{t-quiescent})
if every transaction $T_k \in \ms{txns}(E)$ is complete (resp., t-complete) in $C$.
If a transaction $T$ is incomplete in an execution $E$, it has exactly one \emph{enabled} event, 
which is the next event the transaction will perform according to the TM implementation.
Events $e$ and $e'$ of an execution $E$  \emph{contend} on a base
object $b$ if they are both events on $b$ in $E$ and at least 
one of them is nontrivial (the event is trivial (resp., nontrivial) if it is the application of a trivial (resp., nontrivial) primitive).

We say that $T$ is \emph{poised to apply an event $e$ after $E$} if $e$ is the next enabled event for $T$ in $E$.
We say that transactions $T$ and $T'$ \emph{concurrently contend on $b$ in $E$} 
if they are poised to apply contending events on $b$ after $E$.

We say that an execution fragment $E$ is \emph{step contention-free for t-operation $op_k$} if the events of $E|op_k$ 
are contiguous in $E$.
We say that an execution fragment $E$ is \emph{step contention-free for $T_k$} if the events of $E|k$ are contiguous in $E$.
We say that $E$ is \emph{step contention-free} if $E$ is step contention-free for all transactions that participate in $E$.
\section{TM classes}
\label{sec:prel}
In this section, we define the properties of TM implementations considered in this paper.

\vspace{1mm}\noindent\textbf{TM-correctness.}
Informally, a t-sequential history $S$ is \emph{legal} if every t-read of a t-object returns the latest written value
of this t-object in $S$.
A history $H$ is \emph{opaque}
if there exists a legal t-sequential history $S$ equivalent to $H$ 
such that $S$ respects the real-time order of transactions
in $H$~\cite{tm-book}.
A weaker condition called \emph{strict serializability} ensures opacity only with respect to committed transactions.
Precise definitions can be found in Appendix~\ref{app:upper}.

\vspace{1mm}\noindent\textbf{TM-liveness.}
We say that a TM implementation $M$ provides \emph{obstruction-free (OF) TM-liveness} if for every finite execution $E$ of $M$, 
and every transaction $T_k$ that applies the invocation of a t-operation $op_k$ immediately after $E$, 
the finite step contention-free extension for $op_k$ contains a matching response.
A TM implementation $M$ provides \emph{wait-free TM-liveness} if
in every execution of $M$, every t-operation returns a matching response in a finite number of its steps.

\vspace{1mm}\noindent\textbf{TM-progress.}
Progress for TMs specifies the conditions under which a transaction is allowed to abort.
We say that a TM implementation $M$ provides \emph{obstruction-free (OF) TM-progress} if for every execution $E$ of $M$, 
if any transaction $T_k \in \ms{txns}(E)$ returns $A_k$ in $E$, then $E$ is not step contention-free for $T_k$. 

We say that transactions $T_i,T_j$ \emph{conflict} in an execution $E$ on a t-object $X$ if
$T_i$ and $T_j$ are concurrent in $E$ and $X\in\Dset(T_i)\cap\Dset(T_j)$,  and $X\in\Wset(T_i)\cup\Wset(T_j)$.
A TM implementation $M$ provides \emph{progressive} TM-progress (or \emph{progressiveness}) 
if for every execution $E$ of $M$ and every transaction $T_i \in \ms{txns}(E)$ that returns $A_i$ in $E$, 
there exists 
prefix $E'$ of $E$ and
a transaction $T_k \in \ms{txns}(E')$ such that $T_k$ and $T_i$
conflict in $E$. 

\vspace{1mm}\noindent\textbf{Read invisibility.}
Informally, the invisible reads assumption prevents TM implementations from applying nontrivial primitives
during t-read operations and from announcing read sets of transactions during tryCommit.

We say that a TM implementation $M$ uses \emph{invisible reads} if
for every execution $E$ of $M$,
\begin{itemize}
\item
for every read-only transaction $T_k \in \ms{txns}(E)$, 
no event of $E|k$ is nontrivial in $E$, 
\item
for every updating transaction $T_k \in \ms{txns}(E)$; $\Rset(T_k)\neq \emptyset$, 
there exists an execution $E'$ of $M$ such that
\begin{itemize}
\item
$\Rset(T_k)=\emptyset$ in $E'$
\item
$\ms{txns}(E)=\ms{txns}(E')$ and $\forall T_m \in \ms{txns}(E) \setminus \{T_k\}$: $E|m=E'|m$
\item
for any two step contention-free transactions $T_i, T_j \in \ms{txns}(E)$, 
if the last event of $T_i$ precedes the first event of $T_j$ in $E$, 
then the last event of $T_i$ precedes the first event of $T_j$ in $E'$.
\end{itemize}
\end{itemize}
Most popular TM implementations like \emph{TL2}~\cite{DSS06} and \emph{NOrec}~\cite{norec} satisfy this definition of invisible
reads.

\vspace{1mm}\noindent\textbf{Disjoint-access parallelism (DAP).}
A TM implementation $M$ is \emph{strictly disjoint-access parallel
  (strict DAP)} if, for
all executions $E$ of $M$, and for all transactions $T_i$ and $T_j$ that participate in $E$, 
$T_i$ and $T_j$ contend on a base object in $E$ only if 
$\Dset(T_i)\cap \Dset(T_j)\neq \emptyset$~\cite{tm-book}.

We now describe two relaxations of strict DAP. For the formal definitions, we introduce the notion of a
\emph{conflict graph} which captures the dependency relation among t-objects accessed by transactions.

We denote by $\tau_{E}(T_i,T_j)$, the set of transactions ($T_i$ and $T_j$ included)
that are concurrent to at least one of $T_i$ and $T_j$ in an execution $E$.

Let $G(T_i,T_j,E)$ be an undirected graph whose vertex set is $\bigcup\limits_{T \in \tau_{E}(T_i,T_j)} \Dset(T)$
and there is an edge
between t-objects $X$ and $Y$ \emph{iff} there exists $T \in \tau_{E}(T_i,T_j)$ such that 
$\{X,Y\} \in \Dset(T)$.
We say that $T_i$ and $T_j$ are \emph{disjoint-access} in $E$
if there is no path between a t-object in $\Dset(T_i)$ and a t-object in $\Dset(T_j)$ in $G(T_i,T_j,E)$.
A TM implementation $M$ is \emph{weak disjoint-access parallel (weak DAP)} if, for
all executions $E$ of $M$,
transactions $T_i$ and $T_j$ 
concurrently contend on the same base object in $E$ only if   
$T_i$ and $T_j$ are not disjoint-access in $E$ or there exists a t-object $X \in \Dset(T_i) \cap \Dset(T_j)$~\cite{AHM09,PFK10}.

Let ${\tilde G}(T_i,T_j,E)$ be an undirected graph whose vertex set is $\bigcup_{T \in \tau_{E}(T_i,T_j)} \Dset(T)$
and there is an edge
between t-objects $X$ and $Y$ \emph{iff} there exists $T \in \tau_{E}(T_i,T_j)$ such that 
$\{X,Y\} \in \Wset(T)$.
We say that $T_i$ and $T_j$ are \emph{read-write disjoint-access} in $E$
if there is no path between a t-object in $\Dset(T_i)$ and a t-object in $\Dset(T_j)$ in ${\tilde G}(T_i,T_j,E)$.
A TM implementation $M$ is \emph{read-write disjoint-access parallel (RW DAP)} if, for
all executions $E$ of $M$, 
transactions $T_i$ and $T_j$ 
contend on the same base object in $E$ only if   
$T_i$ and $T_j$ are not read-write disjoint-access in $E$ or there exists a t-object $X \in \Dset(T_i) \cap \Dset(T_j)$.

We make the following observations about the DAP definitions presented in this paper.
\begin{itemize}
\item 
From the definitions, it is immediate that
every RW DAP TM implementation satisfies weak DAP. But the converse is not true.
Consider the following execution $E$ of a weak DAP TM implementaton $M$ 
that begins with the t-incomplete execution of a transaction $T_0$ that 
reads $X$ and writes to $Y$, followed by the step contention-free executions of two transactions $T_1$ and $T_2$ 
which write to $X$ and read $Y$ respectively. Transactions $T_1$ and $T_2$ may contend on a base object since 
there is a path between $X$ and $Y$ in $G(T_1,T_2,E)$. However, a RW DAP TM implementation
would preclude transactions $T_1$ and $T_2$ from contending on the same base object: there is no edge
between t-objects $X$ and $Y$ in the corresponding conflict graph ${\tilde G}(T_1,T_2,E)$ because
$X$ and $Y$ are not contained in the write set of $T_0$.
Algorithm~\ref{alg:oftm2} in Appendix~\ref{app:woftm} describes a TM implementation that satisfies weak DAP, but not RW DAP.
\item 
From the definitions, it is immediate that
every strict DAP TM implementation satisfies RW DAP. But the converse is not true.
To understand why,
consider the following execution $E$ of a RW DAP TM implementaton
that begins with the t-incomplete execution of a transaction $T_0$ that 
accesses t-objects $X$ and $Y$, followed by the step contention-free executions of two transactions $T_1$ and $T_2$ 
which access $X$ and $Y$ respectively. Transactions $T_1$ and $T_2$ may contend on a base object since 
there is a path between $X$ and $Y$ in ${\tilde G}(T_i,T_j,E)$. 
However, a strict DAP TM implementation
would preclude transactions $T_1$ and $T_2$ from contending on the same base object since $\Dset(T_1) \cap \Dset(T_2)=\emptyset$
in $E$. Algorithm~\ref{alg:oftm} in Appendix~\ref{app:rwoftm} describes a TM implementation that satisfies RW DAP, but not strict DAP.
\end{itemize}
%
%
%
%
%
%
%
%
\section{Lower bounds for obstruction-free TMs}
\label{sec:oftm}
Let $\mathcal{OF}$ denote the class of TMs that  provide OF TM-progress and OF TM-liveness.
In Section~\ref{sec:ofinv}, we show that no strict serializable TM in $\mathcal{OF}$ can be
weak DAP and have invisible reads.
In Section~\ref{sec:ofsrt}, we determine stall complexity bounds for
strict serializable TMs in $\mathcal{OF}$, and in Section~\ref{sec:oflowerraw}, we
present a linear (in $n$) lower bound on RAW/AWARs for RW DAP opaque TMs in
$\mathcal{OF}$. 
%
\subsection{Impossibility of invisible reads}
\label{sec:ofinv}
In this section, we prove that it is impossible to derive TM implementations in $\mathcal{OF}$ that combine
weak DAP and invisible reads.
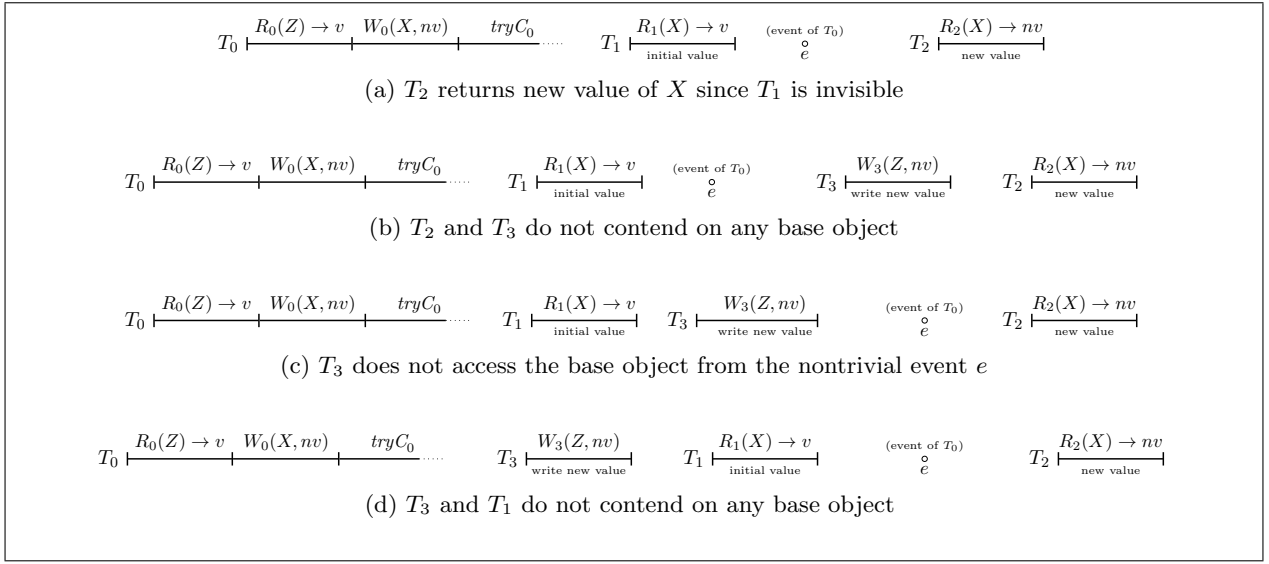
\begin{figure*}[t]
\begin{center}
	\subfloat[$T_2$ returns new value of $X$ since $T_1$ is invisible\label{sfig:inv-1}]{\scalebox{0.7}[0.7]{\begin{tikzpicture}
\node (r1) at (1,0) [] {};
\node (w1) at (3,0) [] {};
\node (c1) at (5,0) [] {};

\node (r2) at (8.2,0) [] {};

\node (e) at (10.5,0) [] {};

\node (r3) at (14,0) [] {};

\draw (r1) node [above] {\small {$R_0(Z)\rightarrow v$}};
\draw (w1) node [above] {\small {$W_0(X,nv)$}};
\draw (c1) node [above] {\small {$\TryC_0$}};

\draw (r2) node [above] {\small {$R_1(X)\rightarrow v$}};
\draw (r2) node [below] {\tiny {initial value}};

\draw (e) node [above] {\tiny {(event of $T_0$)}};
\draw (e) node [below] {\small {$e$}};

\draw (r3) node [above] {\small {$R_2(X)\rightarrow nv$}};
\draw (r3) node [below] {\tiny {new value}};

\begin{scope}   
\draw [|-|,thick] (0,0) node[left] {$T_0$} to (2,0);
\draw [-|,thick] (2,0) node[left] {} to (4,0);
\draw [-,thick] (4,0) node[left] {} to (5.5,0);
\draw [-,dotted] (5.5,0) node[left] {} to (6,0);
\draw [|-|,thick] (7.2,0) node[left] {$T_1$} to (9.2,0);
\draw  (10.5,0) circle [fill, radius=0.05]  (10.5,0);
\draw [|-|,thick] (13,0) node[left] {$T_2$} to (15,0);
\end{scope}
\end{tikzpicture}}}
        \\
        \vspace{2mm}
	\subfloat[$T_2$ and $T_3$ do not contend on any base object\label{sfig:inv-2}]{\scalebox{0.7}[0.7]{\begin{tikzpicture}
\node (r1) at (1,0) [] {};
\node (w1) at (3,0) [] {};
\node (c1) at (5,0) [] {};

\node (r2) at (8.2,0) [] {};

\node (e) at (10.5,0) [] {};

\node (r3) at (17.5,0) [] {};

\node (w3) at (14,0) [] {};

\draw (r1) node [above] {\small {$R_0(Z)\rightarrow v$}};
\draw (w1) node [above] {\small {$W_0(X,nv)$}};
\draw (c1) node [above] {\small {$\TryC_0$}};

\draw (r2) node [above] {\small {$R_1(X)\rightarrow v$}};
\draw (r2) node [below] {\tiny {initial value}};

\draw (e) node [above] {\tiny {(event of $T_0$)}};
\draw (e) node [below] {\small {$e$}};

\draw (r3) node [above] {\small {$R_2(X)\rightarrow nv$}};
\draw (r3) node [below] {\tiny {new value}};

\draw (w3) node [above] {\small {$W_3(Z,nv)$}};
\draw (w3) node [below] {\tiny {write new value}};

\begin{scope}   
\draw [|-|,thick] (0,0) node[left] {$T_0$} to (2,0);
\draw [-|,thick] (2,0) node[left] {} to (4,0);
\draw [-,thick] (4,0) node[left] {} to (5.5,0);
\draw [-,dotted] (5.5,0) node[left] {} to (6,0);
\draw [|-|,thick] (7.2,0) node[left] {$T_1$} to (9.2,0);
\draw  (10.5,0) circle [fill, radius=0.05]  (10.5,0);
\draw [|-|,thick] (13,0) node[left] {$T_3$} to (15,0);
\draw [|-|,thick] (16.5,0) node[left] {$T_2$} to (18.5,0);
\end{scope}
\end{tikzpicture}}}
	\\
	\vspace{2mm}
	\subfloat[$T_3$ does not access the base object from the nontrivial event $e$\label{sfig:inv-3}]{\scalebox{0.7}[0.7]{\begin{tikzpicture}
\node (r1) at (1,0) [] {};
\node (w1) at (3,0) [] {};
\node (c1) at (5,0) [] {};

\node (r2) at (8.2,0) [] {};

\node (w3) at (11.5,0) [] {};

\node (r3) at (17.5,0) [] {};

\node (e) at (14.5,0) [] {};

\draw (r1) node [above] {\small {$R_0(Z)\rightarrow v$}};
\draw (w1) node [above] {\small {$W_0(X,nv)$}};
\draw (c1) node [above] {\small {$\TryC_0$}};

\draw (r2) node [above] {\small {$R_1(X)\rightarrow v$}};
\draw (r2) node [below] {\tiny {initial value}};

\draw (e) node [above] {\tiny {(event of $T_0$)}};
\draw (e) node [below] {\small {$e$}};

\draw (r3) node [above] {\small {$R_2(X)\rightarrow nv$}};
\draw (r3) node [below] {\tiny {new value}};

\draw (w3) node [above] {\small {$W_3(Z,nv)$}};
\draw (w3) node [below] {\tiny {write new value}};

\begin{scope}   
\draw [|-|,thick] (0,0) node[left] {$T_0$} to (2,0);
\draw [-|,thick] (2,0) node[left] {} to (4,0);
\draw [-,thick] (4,0) node[left] {} to (5.5,0);
\draw [-,dotted] (5.5,0) node[left] {} to (6,0);
\draw [|-|,thick] (7.1,0) node[left] {$T_1$} to (9.1,0);
\draw  (14.5,0) circle [fill, radius=0.05]  (14.5,0);
\draw [|-|,thick] (10.2,0) node[left] {$T_3$} to (12.5,0);
\draw [|-|,thick] (16.5,0) node[left] {$T_2$} to (18.5,0);
\end{scope}
\end{tikzpicture}}}
	\\
	\vspace{2mm}
	\subfloat[$T_3$ and $T_1$ do not contend on any base object \label{sfig:inv-4}]{\scalebox{0.7}[0.7]{\begin{tikzpicture}
\node (r1) at (1,0) [] {};
\node (w1) at (3,0) [] {};
\node (c1) at (5,0) [] {};

\node (r2) at (12,0) [] {};

\node (e) at (15,0) [] {};

\node (r3) at (18.5,0) [] {};

\node (w3) at (8.5,0) [] {};

\draw (r1) node [above] {\small {$R_0(Z)\rightarrow v$}};
\draw (w1) node [above] {\small {$W_0(X,nv)$}};
\draw (c1) node [above] {\small {$\TryC_0$}};

\draw (r2) node [above] {\small {$R_1(X)\rightarrow v$}};
\draw (r2) node [below] {\tiny {initial value}};

\draw (e) node [above] {\tiny {(event of $T_0$)}};
\draw (e) node [below] {\small {$e$}};

\draw (r3) node [above] {\small {$R_2(X)\rightarrow nv$}};
\draw (r3) node [below] {\tiny {new value}};

\draw (w3) node [above] {\small {$W_3(Z,nv)$}};
\draw (w3) node [below] {\tiny {write new value}};

\begin{scope}   
\draw [|-|,thick] (0,0) node[left] {$T_0$} to (2,0);
\draw [-|,thick] (2,0) node[left] {} to (4,0);
\draw [-,thick] (4,0) node[left] {} to (5.5,0);
\draw [-,dotted] (5.5,0) node[left] {} to (6,0);
\draw [|-|,thick] (7.5,0) node[left] {$T_3$} to (9.5,0);
\draw [|-|,thick] (11,0) node[left] {$T_1$} to (13,0);
\draw  (15,0) circle [fill, radius=0.05]  (15,0);
\draw [|-|,thick] (17.5,0) node[left] {$T_2$} to (19.5,0);
\end{scope}
\end{tikzpicture}}}
	\caption{Executions in the proof of Theorem~\ref{th:ir}; execution in \ref{sfig:inv-4} is not strictly serializable
          \label{fig:indis}} 
\end{center}
\end{figure*}
The following lemma will be useful in proving our result.
\begin{lemma}
\label{lm:dap}
(\cite{AHM09},\cite{WF14-icdcn})
Let $M $ be any weak DAP TM implementation.
Let $\alpha\cdot \rho_1 \cdot \rho_2$ be any execution of $M$ where
$\rho_1$ (resp., $\rho_2$) is the step contention-free
execution fragment of transaction $T_1 \not\in \ms{txns}(\alpha)$ (resp., $T_2 \not\in \ms{txns}(\alpha)$) 
and transactions $T_1$, $T_2$ are disjoint-access in $\alpha\cdot \rho_1 \cdot \rho_2$. 
Then, $T_1$ and $T_2$ do not contend on any base object in $\alpha\cdot \rho_1 \cdot \rho_2$.
\end{lemma}
\begin{theorem}
\label{th:ir}
There does not exist a weak DAP strictly serializable TM
implementation in $\mathcal{OF}$ that uses invisible reads.
\end{theorem}
\begin{proof}
By contradiction, assume that such an implementation $M\in\mathcal{OF}$ exists.
Let $v$ be the initial value of t-objects $X$ and $Z$.
Consider an execution $E$ of $M$ 
in which 
a transaction $T_0$ performs $\Read_0(Z) \rightarrow v$ (returning $v$), 
writes $nv\neq v$ to $X$, and commits.
Let $E'$ denote the longest prefix of $E$ that cannot be extended with
the t-complete step contention-free execution of transaction $T_1$ that performs 
a t-read $X$ and returns $nv$ nor with the t-complete step contention-free execution of transaction $T_2$ that
performs a t-read of $X$ and returns $nv$. 

Let $e$ be the enabled event of transaction $T_0$ in the configuration after $E'$.
Without loss of generality, assume that $E'\cdot e$ can be extended with the t-complete step contention-free
execution of committed transaction $T_2$ that reads $X$ and returns $nv$.
Let $E' \cdot e \cdot E_2$ be such an execution, where 
$E_2$ is the t-complete step contention-free execution fragment of transaction $T_2$ that 
performs $\Read_2(X) \rightarrow nv$ and commits.

We now prove that $M$ has an execution of the form $E' \cdot E_1 \cdot e \cdot E_2$, where
$E_1$ is the t-complete step contention-free execution fragment of transaction $T_1$ that 
performs $\Read_1(X) \rightarrow v$ and commits.

We observe that $E'\cdot E_1$ is an execution of $M$.
Indeed, by OF TM-progress and OF TM-liveness, $T_1$ must return a matching response that is not $A_1$ in $E'\cdot E_1$,
and by the definition of $E'$, this response must be the initial value $v$ of $X$.

By the assumption of invisible reads,
$E_1$ does not contain any nontrivial events.
Consequently, $E'\cdot E_1 \cdot e \cdot E_2$ is indistinguishable to transaction $T_2$
from the execution $E'\cdot e  \cdot E_2$. 
Thus, $E'\cdot E_1 \cdot e \cdot E_2$ is also an execution of $M$ (Figure~\ref{sfig:inv-1}).
\begin{claim}
\label{cl:invdapof2}
$M$ has an execution of the form $E'\cdot E_1 \cdot  E_3  \cdot e \cdot E_2$ where
$E_3$ is the t-complete step contention-free execution fragment of transaction $T_{3}$ that 
writes $nv \neq v$ to $Z$ and commits.
\end{claim}
\begin{proof}
The proof is through a sequence of indistinguishability arguments to construct the execution.

We first claim that $M$ has an execution of the form $E' \cdot E_1 \cdot e \cdot E_2 \cdot E_3$.
Indeed, by OF TM-progress and OF TM-liveness, $T_3$ must be committed in $E' \cdot E_1 \cdot e \cdot E_2 \cdot E_3$.

Since $M$ uses invisible reads,
the execution $E' \cdot E_1 \cdot e \cdot E_2 \cdot E_3$ is indistinguishable
to transactions $T_2$ and $T_3$ from the execution ${\hat E}\cdot E_2 \cdot E_3$, where ${\hat E}$ is
the t-incomplete step contention-free execution of transaction $T_0$ with
$\Wset_{\hat E}(T_0)=\{X\}$; $\Rset_{\hat E}(T_0)=\emptyset$ that writes $nv$ to $X$.

Observe that the execution $E'\cdot E_1 \cdot e \cdot E_2 \cdot E_3$ is indistinguishable
to transactions $T_2$ and $T_3$ from the execution ${\hat E} \cdot E_2 \cdot E_3$, in which
transactions $T_3$ and $T_2$ are disjoint-access. Consequently, by Lemma~\ref{lm:dap}, $T_2$ and $T_3$ 
do not contend on any base object in ${\hat E}  \cdot E_2 \cdot E_3$.
Thus, $M$ has an execution of the form $E' \cdot E_1 \cdot e \cdot E_3 \cdot E_2$ (Figure~\ref{sfig:inv-2}).

By definition of $E'$, $T_0$ applies a nontrivial primitive to some base object, say $b$, in event $e$ that
$T_2$ must access in $E_2$.
Thus, the execution fragment $E_3$ does not contain any nontrivial event on $b$ in
the execution $E'\cdot E_1 \cdot e \cdot E_2 \cdot E_3$.
Infact, since $T_3$ is disjoint-access with $T_0$ in the execution ${\hat E} \cdot  E_3 \cdot E_2$,
by Lemma~\ref{lm:dap}, it cannot access the base object $b$ to which $T_0$ applies a nontrivial primitive
in the event $e$. Thus, transaction
$T_3$ must perform the same sequence of events $E_3$ immediately after $E'$, implying that
$M$ has an execution of the form 
$E'\cdot E_1 \cdot  E_3 \cdot e \cdot E_2$ (Figure~\ref{sfig:inv-3}).
\end{proof} 
Finally, we observe that the execution $E' \cdot E_1 \cdot E_3 \cdot e \cdot E_2$ established in Claim~\ref{cl:invdapof2}
is indistinguishable
to transactions $T_1$ and $T_3$ from an execution ${\tilde E}\cdot E_1 \cdot E_3 \cdot e \cdot E_2$, 
where $\Wset(T_0)=\{X\}$ and $\Rset(T_0)=\emptyset$ in $\tilde E$.
But transactions $T_3$ and $T_1$ are disjoint-access in ${\tilde E} \cdot  E_1 \cdot E_3 \cdot e \cdot E_2$
and by Lemma~\ref{lm:dap}, $T_1$ and $T_3$ do not contend on any base object in this execution.
Thus, $M$ has an execution of the form $E' \cdot E_3 \cdot E_1 \cdot e \cdot E_2$ (Figure~\ref{sfig:inv-4})
in which $T_3$ precedes $T_1$ in real-time order.

However, the execution $E' \cdot E_3 \cdot E_1 \cdot e \cdot E_2$ is not strictly serializable:
$T_0$ must be committed in any serialization and transaction $T_1$
must precede $T_0$ since $\Read_1(X)$ returns the initial value of $X$. 
To respect real-time order, $T_3$ must precede $T_1$, while $T_0$ must precede $T_2$ since 
$\Read_2(X)$ returns $nv$, the value of $X$ updated by $T_0$.
Finally, $T_0$ must precede $T_3$ since $\Read_0(Z)$ returns the initial value of $Z$.
But there exists no such serialization---contradiction.
\end{proof}
%
%
%
%
%

%
\subsection{Stall complexity}
\label{sec:ofsrt}
Let $M$ be any TM implementation.
Let $e$ be an event applied by process $p$ to a base object $b$ as it performs a transaction $T$ during an execution $E$ of $M$.
Let $E=\alpha\cdot e_1\cdots e_m \cdot e \cdot \beta$ be an execution of $M$, where $\alpha$ and $\beta$ are execution 
fragments and $e_1\cdots e_m$
is a maximal sequence of $m\geq 1$ consecutive nontrivial events by distinct distinct processes other than $p$ that access $b$.
Then, we say that $T$ incurs $m$ \emph{memory stalls in $E$ on account of $e$}.
The \emph{number of memory stalls incurred by $T$ in $E$} is the sum of memory stalls incurred by all events of $T$ in $E$~\cite{G05,AGHK09}.

In this section, we prove
a lower bound of $n-1$ on the worst case number of stalls incurred by a transaction as it performs a single t-read operation.
We adopt the following definition of a \emph{k-stall execution} from \cite{AGHK09,G05}.
\begin{definition}
\label{def:stalls}
An execution $\alpha\cdot \sigma_1 \cdots \sigma_i$ is a $k$-stall execution for t-operation $op$ executed by process $p$ if
\begin{itemize}
\item 
$\alpha$ is $p$-free,
\item
there are distinct base objects $b_1,\ldots , b_i$ and disjoint sets of processes $S_1,\ldots , S_i$
whose union does not include $p$
and has cardinality $k$ such that, for $j=1,\ldots i $,
\begin{itemize}
\item
each process in $S_j$ has an enabled nontrivial event about to access base object $b_j$ after $\alpha$, and
\item
in $\sigma_j$, $p$ applies events by itself until it is the first about to apply an event to $b_j$,
then each of the processes in $S_j$ applies an event that accesses $b_j$, and finally, $p$ applies an event that accesses $b_j$,
\end{itemize}
\item
$p$ invokes exactly one t-operation $op$ in the execution fragment $\sigma_1\cdots \sigma_i$
\item
$\sigma_1\cdots \sigma_i$ contains no events of processes not in $(\{p\}\cup S_1\cup \cdots \cup S_i)$
\item
in every $(\{p\}\cup S_1\cup \cdots \cup S_i)$-free execution fragment that extends $\alpha$, 
no process applies a nontrivial event to any base object accessed in $\sigma_1 \cdots \sigma_i$.
\end{itemize}
\end{definition}
Observe that in a $k$-stall execution $E$ for t-operation $op$, the number of memory stalls incurred by $op$
in $E$ is $k$.
\begin{lemma}
\label{lm:stalls}
Let $\alpha\cdot \sigma_1 \cdots \sigma_i$ be a $k$-stall execution for t-operation $op$ executed by process $p$.
Then, $\alpha\cdot \sigma_1 \cdots \sigma_i$ is indistinguishable to $p$ from a step contention-free execution~\cite{AGHK09}.
\end{lemma}
\begin{theorem}
\label{th:oftmstalls} 
Every strictly serializable TM implementation $M\in \mathcal{OF}$ has a $(n-1)$-stall execution $E$ for a t-read operation
performed in $E$.
\end{theorem}
\begin{proof}
We proceed by induction.
Observe that the empty
execution is a $0$-stall execution since it vacuously satisfies the invariants of Definition~\ref{def:stalls}. 

Let $v$ be the initial value of t-objects $X$ and $Z$.
Let $\alpha={\alpha}_1\cdots { \alpha}_{n-2}$ be a step contention-free execution of a strictly serializable TM implementation
$M\in \mathcal{OF}$, where for all $j\in \{1,\ldots, n-2\}$,
$\alpha_j$ is the longest prefix of the execution fragment ${\bar \alpha}_j$ that denotes the t-complete step-contention
free execution of committed transaction $T_j$ (invoked by process $p_j$)
that performs $\Read_j(Z)\rightarrow v$, writes value $nv \neq v$ to $X$ 
in the execution ${\alpha}_1\cdots { \alpha}_{j-1} \cdot {\bar \alpha}_j$
such that
\begin{itemize}
\item 
$\TryC_j()$ is incomplete in $\alpha_j$,
\item
$\alpha_1\cdots \alpha_{j}$ cannot be extended with the t-complete step contention-free execution fragment
of any transaction $T_{n-1}$ or $T_n$ that performs exactly one t-read of $X$ that returns $nv$ and commits.
\end{itemize}
Assume, inductively, that $\alpha \cdot \sigma_1\cdots \sigma_i$ is a $k$-stall execution
for $\Read_n(X)$ executed by process $p_n$, where $0\leq k \leq n-2$. 
By Definition~\ref{def:stalls}, there are distinct base objects $b_1,\ldots b_i$ 
accessed by disjoint sets of processes $S_1\ldots S_i$ in the execution fragment $\sigma_1\cdots \sigma_i$,
where $|S_1\cup \ldots \cup S_i|=k$ and 
$\sigma_1\cdots \sigma_i$ contains no events of processes not in $S_1 \cup \ldots  \cup S_i \cup \{p_n\}$.
We will prove that there exists a $(k+k')$-stall execution for $\Read_n(X)$, for some $k'\geq 1$.

By Lemma~\ref{lm:stalls}, 
$\alpha \cdot \sigma_1\cdots \sigma_i$ is indistinguishable to $T_n$ from a step contention-free execution.
Let $\sigma$ be the finite step contention-free execution fragment that extends $\alpha \cdot \sigma_1\cdots \sigma_i$
in which $T_n$ performs events by itself: completes $\Read_n(X)$ and returns a response. 
By OF TM-progress and OF TM-liveness, $\Read_n(X)$ and the subsequent $\TryC_k$ must each
return non-$A_n$ responses in $\alpha \cdot \sigma_1\cdots \sigma_i\cdot \sigma$.
By construction of $\alpha$ and strict serializability of $M$, $\Read_n(X)$
must return the response $v$ or $nv$ in this execution. 
We prove that there exists an execution fragment $\gamma$ performed 
by some process $p_{n-1} \not\in (\{p_n\}\cup S_1\cup \cdots \cup S_i)$
extending $\alpha$ that contains a nontrivial event on some
base object that must be accessed by $\Read_n(X)$ in $\sigma_1\cdots \sigma_i\cdot \sigma$.

Consider the case that $\Read_n(X)$
returns the response $nv$ in $\alpha \cdot \sigma_1\cdots \sigma_i\cdot \sigma$.
We define a step contention-free fragment $\gamma$ extending $\alpha$ that is the t-complete step contention-free 
execution of transaction $T_{n-1}$ executed by some process $p_{n-1} \not\in (\{p_n\}\cup S_1\cup \cdots \cup S_i)$ 
that performs $\Read_{n-1}(X)\rightarrow v$, writes 
$nv\neq v$ to $Z$ and commits.
By definition of $\alpha$, OF TM-progress and OF TM-liveness, $M$ has an execution of the form
$\alpha\cdot \gamma$.
We claim that the execution fragment $\gamma$ must contain a nontrivial event on some base object that
must be accessed by $\Read_n(X)$ in $\sigma_1\cdots \sigma_i\cdot \sigma$.
Suppose otherwise. 
Then, $\Read_n(X)$ must return the response $nv$ in $\sigma_1\cdots \sigma_i\cdot \sigma$.
But the execution $\alpha\cdot \sigma_1\cdots \sigma_i\cdot \sigma$ is not strictly serializable.
Since $\Read_n(X)\rightarrow nv$, there exists a transaction $T_q\in \ms{txns}(\alpha)$ that must be committed
and must precede $T_n$ in any serialization.
Transaction $T_{n-1}$ must precede $T_n$ in any serialization to respect the real-time order and
$T_{n-1}$ must precede $T_q$ in any serialization. Also, $T_q$ must precede $T_{n-1}$ in any serialization.
But there exists no such serialization.

Consider the case that $\Read_n(X)$
returns the response $v$ in $\alpha \cdot \sigma_1\cdots \sigma_i\cdot \sigma$.
In this case, we define the step contention-free fragment $\gamma$ extending $\alpha$ as the t-complete step contention-free 
execution of transaction $T_{n-1}$ executed by some process $p_{n-1} \not\in (\{p_n\}\cup S_1\cup \cdots \cup S_i)$ 
that writes $nv\neq v$ to $X$ and commits.
By definition of $\alpha$, OF TM-progress and OF TM-liveness, $M$ has an execution of the form
$\alpha\cdot \gamma$.
By strict serializability of $M$, the execution fragment $\gamma$ must contain a nontrivial event on some base object that
must be accessed by $\Read_n(X)$ in $\sigma_1\cdots \sigma_i\cdot \sigma$. Suppose otherwise.
Then, $\sigma_1\cdots \sigma_i \cdot \gamma \cdot \sigma$ is an execution of $M$
in which $\Read_n(X)\rightarrow v$. But this execution is not strictly
serializable: every transaction $T_q\in \ms{txns}(\alpha)$ must be aborted or must be preceded by $T_n$
in any serialization, but committed transaction $T_{n-1}$ must precede $T_n$ in any serialization to
respect the real-time ordering of transactions. But then $\Read_n(X)$ must return the new value $nv$ of $X$ that is
updated by $T_{n-1}$---contradiction.

Since, by Definition~\ref{def:stalls},
the execution fragment $\gamma$ executed by some process $p_{n-1} \not\in (\{p_n\}\cup S_1\cup \cdots \cup S_i)$
contains no nontrivial events to any base object accessed in $\sigma_1 \cdots \sigma_i$,
it must contain a nontrivial event to some base object $b_{i+1}\not\in \{b_1,\ldots , b_i\}$ that is
accessed by $T_n$ in the execution fragment $\sigma$.

Let $\mathcal{A}$ denote the set of all finite $(\{p_n\}\cup S_1 \ldots \cup S_i)$-free execution fragments that extend $\alpha$.
Let $b_{i+1} \not\in \{b_1,\ldots , b_i\}$ be the first base object accessed by $T_n$ in the execution fragment
$\sigma$ to which some transaction applies a nontrivial event in the 
execution fragment $\alpha'\in \mathcal{A}$.
Clearly, some such execution $\alpha \cdot \alpha'$ exists that contains a nontrivial event in $\alpha'$ to some
distinct base object $b_{i+1}$ not accessed in the execution fragment $\sigma_1 \cdots \sigma_i$.
We choose the execution $\alpha\cdot \alpha' \in \mathcal{A}$ that maximizes the number of transactions
that are poised to apply nontrivial events on $b_{i+1}$ in the configuration after $\alpha\cdot \alpha'$.
Let $S_{i+1}$ denote the set of processes executing these transactions and $k'=|S_{i+1}|$ ($k'>0$ as already proved).

We now construct a $(k+k')$-stall execution
$\alpha\cdot \alpha' \cdot \sigma_1 \cdots \sigma_i \cdot \sigma_{i+1}$ for $\Read_n(X)$,
where in $\sigma_{i+1}$, $p_n$ applies events by itself, then each of the processes in $S_{i+1}$ applies a nontrivial event
on $b_{i+1}$, and finally, $p_n$ accesses $b_{i+1}$.

By construction, $\alpha\cdot \alpha'$ is $p_n$-free.
Let $\sigma_{i+1}$ be the prefix of $\sigma$ not including $T_n$'s first access to $b_{i+1}$, concatenated with
the nontrivial events on $b_{i+1}$ by each of the $k'$ transactions executed by processes in $S_{i+1}$ 
followed by the access of $b_{i+1}$
by $T_n$. Observe that $T_n$ performs exactly one t-operation $\Read_n(X)$ in the execution fragment $\sigma_1\cdots \sigma_{i+1}$
and $\sigma_1\cdots \sigma_{i+1}$ contains no events of processes not in $(\{p_n\}\cup S_1\cup \cdots \cup S_i \cup S_{i+1})$.

To complete the induction, we need to show that in every $(\{p_n\}\cup S_1\cup \cdots \cup S_i \cup S_{i+1})$-free extension of 
$\alpha \cdot \alpha'$,
no transaction applies a nontrivial event to any base object accessed 
in the execution fragment $\sigma_1 \cdots \sigma_i \cdot \sigma_{i+1}$.
Let $\beta$ be any such execution fragment that extends $\alpha\cdot \alpha' $.
By our construction, $\sigma_{i+1}$ is the execution fragment that consists of events by $p_n$ on base objects accessed in
$\sigma_1\cdots \sigma_i$, nontrivial events on $b_{i+1}$ by transactions in $S_{i+1}$ and finally, an access to $b_{i+1}$
by $p_n$. 
Since $\alpha\cdot \sigma_1\cdots \sigma_i$ is a $k$-stall execution by our induction hypothesis,
$\alpha'\cdot \beta$ is $(\{p_n\}\cup S_1 \ldots \cup S_i\})$-free and thus, $\alpha'\cdot \beta$
does not contain nontrivial events on any base object accessed in $\sigma_1\cdots \sigma_i$.
We now claim that $\beta$ does not contain nontrivial events to $b_{i+1}$. Suppose otherwise.
Thus, there exists some transaction $T'$ that has an enabled nontrivial event to $b_{i+1}$ in the
configuration after $\alpha\cdot \alpha' \cdot \beta'$, where $\beta'$ is some prefix of $\beta$.
But this contradicts the choice of $\alpha \cdot \alpha'$ as the extension of $\alpha$ that maximizes $k'$.

Thus, $\alpha\cdot \alpha'\cdot \sigma_1 \cdots \sigma_i \cdot \sigma_{i+1}$ is indeed a $(k+k')$-stall execution for $T_n$
where $1< k< (k+k') \leq (n-1)$. 
%
\end{proof}
%
%

%
\subsection{RAW/AWAR complexity}
\label{sec:oflowerraw}
Attiya \emph{et al.}~\cite{AGK11-popl} identified two common expensive synchronization patterns
that frequently arise in the design of concurrent algorithms: 
\emph{read-after-write (RAW)} and \emph{atomic write-after-read (AWAR)}.
In this section, we prove that opaque, RW DAP TM implementations in $\mathcal{OF}$ have executions in which
some read-only transaction performs a linear (in $n$) number of RAWs or AWARs.

We recall the formal definitions of RAW and AWAR from \cite{AGK11-popl}.
Let $\pi^i$ denote the $i$-th
event in an execution $\pi$ ($i=0,\ldots , |\pi|-1$).   

We say that a transaction $T$ performs a \emph{RAW} (read-after-write) in $\pi$ if  
$\exists i,j; 0 \leq i < j < |\pi|$
such that (1) $\pi^i$ is a write to a base object $b$ by $T$, 
(2) $\pi^j$ is a read of a base object $b'\neq b$ by $T$ and 
(3) there is no $\pi^k$ such that $i<k<j$ and $\pi^k$ is a write to $b'$ by $T$.
In this paper, we are concerned only with \emph{non-overlapping} RAWs, \emph{i.e.}, the read performed by one precedes the write 
performed by the other.

We say a transaction $T$ performs an \emph{AWAR} (atomic-write-after-read)
in $\pi$ if $\exists i, 0 \leq i < |\pi|$ such that the event $\pi^i$
is the application of a nontrivial primitive that atomically reads a base object $b$ and writes to $b$.
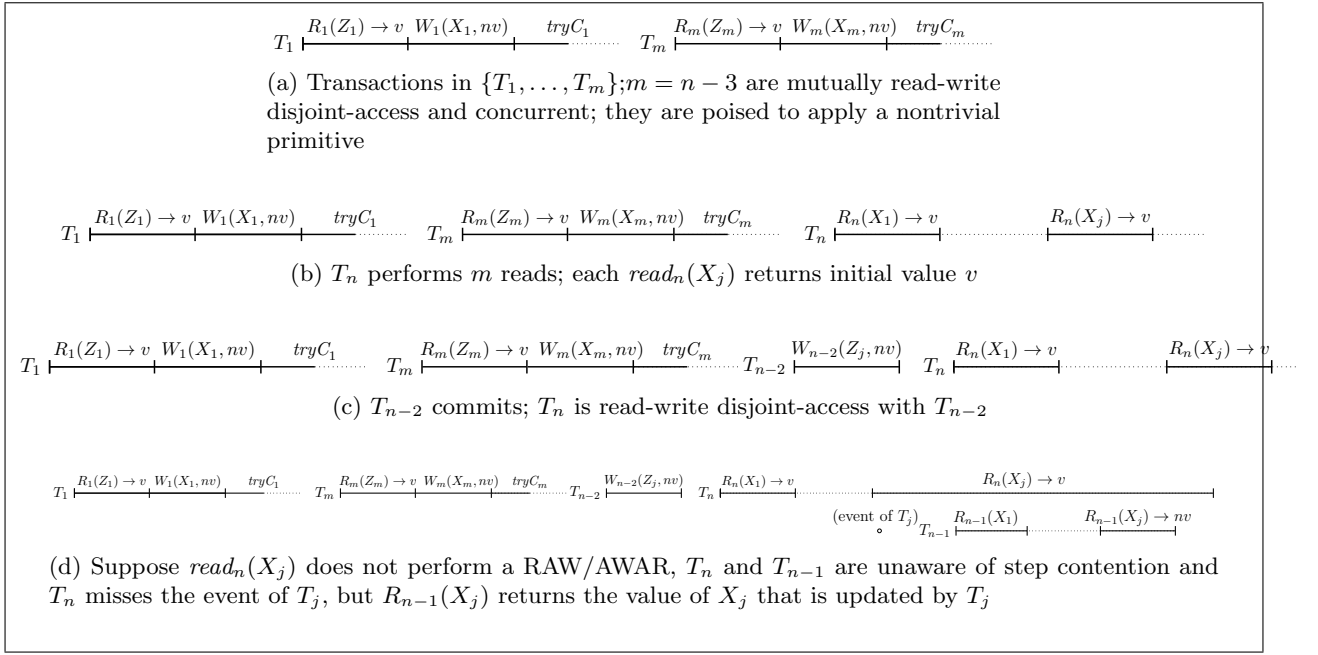
\begin{figure*}[t]
\begin{center}
	\subfloat[Transactions in $\{T_1,\ldots , T_m\}$;$m=n-3$ are mutually read-write disjoint-access and concurrent; they are
	poised to apply a nontrivial primitive\label{sfig:rw-1}]{\scalebox{0.7}[0.7]{\begin{tikzpicture}
\node (r1) at (1,0) [] {};
\node (w1) at (3,0) [] {};
\node (c1) at (5,0) [] {};

\draw (r1) node [above] {\small {$R_1(Z_1)\rightarrow v$}};
\draw (w1) node [above] {\small {$W_1(X_1,nv)$}};
\draw (c1) node [above] {\small {$\TryC_1$}};


\node (rj) at (8,0) [] {};
\node (wj) at (10.2,0) [] {};
\node (cj) at (12,0) [] {};

\draw (rj) node [above] {\small {$R_m(Z_m)\rightarrow v$}};
\draw (wj) node [above] {\small {$W_m(X_m,nv)$}};
\draw (cj) node [above] {\small {$\TryC_m$}};


\begin{scope}   
\draw [|-,thick] (0,0) node[left] {$T_1$} to (5,0);
\draw [|-|,thick] (0,0) node[left] {} to (2,0);
\draw [-|,thick] (2,0) node[left] {} to (4,0);
\draw [-,dotted] (5,0) node[left] {} to (6,0);
\end{scope}
\begin{scope}   
\draw [|-|,thick] (7,0) node[left] {$T_m$} to (9,0);
\draw [-|,thick] (9,0) node[left] {} to (11,0);
\draw [-,thick] (11,0) node[left] {} to (12,0);
\draw [-,dotted] (11,0) node[left] {} to (13,0);
\end{scope}

\end{tikzpicture}}}
        \\
        \vspace{2mm}
        \subfloat[$T_{n}$ performs $m$ reads; each $\Read_{n}(X_j)$ returns initial value $v$\label{sfig:rw-2}]{\scalebox{0.7}[0.7]{\begin{tikzpicture}
\node (r1) at (1,0) [] {};
\node (w1) at (3,0) [] {};
\node (c1) at (5,0) [] {};

\draw (r1) node [above] {\small {$R_1(Z_1)\rightarrow v$}};
\draw (w1) node [above] {\small {$W_1(X_1,nv)$}};
\draw (c1) node [above] {\small {$\TryC_1$}};


\node (rj) at (8,0) [] {};
\node (wj) at (10.2,0) [] {};
\node (cj) at (12,0) [] {};

\draw (rj) node [above] {\small {$R_m(Z_m)\rightarrow v$}};
\draw (wj) node [above] {\small {$W_m(X_m,nv)$}};
\draw (cj) node [above] {\small {$\TryC_m$}};


\node (p1) at (15,0) [] {};
\node (pj) at (19,0) [] {};

\draw (p1) node [above] {\small {$R_{n}(X_1)\rightarrow v$}};
\draw (pj) node [above] {\small {$R_{n}(X_j)\rightarrow v$}};

\begin{scope}   
\draw [|-,thick] (0,0) node[left] {$T_1$} to (5,0);
\draw [|-|,thick] (0,0) node[left] {} to (2,0);
\draw [-|,thick] (2,0) node[left] {} to (4,0);
\draw [-,dotted] (5,0) node[left] {} to (6,0);
\end{scope}
\begin{scope}   
\draw [|-|,thick] (7,0) node[left] {$T_m$} to (9,0);
\draw [-|,thick] (9,0) node[left] {} to (11,0);
\draw [-,thick] (11,0) node[left] {} to (12,0);
\draw [-,dotted] (11,0) node[left] {} to (13,0);
\end{scope}
\begin{scope}   
\draw [|-,dotted] (14,0) node[left] {$T_{n}$} to (21,0);
\draw [|-|,thick] (14,0) to (16,0);
\draw [|-|,thick] (18,0) to (20,0);
\end{scope}
\end{tikzpicture}}}
        \\
        \vspace{2mm}
        \subfloat[$T_{n-2}$ commits; $T_{n}$ is read-write disjoint-access with $T_{n-2}$\label{sfig:rw-3}]{\scalebox{0.7}[0.7]{\begin{tikzpicture}
\node (r1) at (1,0) [] {};
\node (w1) at (3,0) [] {};
\node (c1) at (5,0) [] {};

\draw (r1) node [above] {\small {$R_1(Z_1)\rightarrow v$}};
\draw (w1) node [above] {\small {$W_1(X_1,nv)$}};
\draw (c1) node [above] {\small {$\TryC_1$}};


\node (rj) at (8,0) [] {};
\node (wj) at (10.2,0) [] {};
\node (cj) at (12,0) [] {};

\draw (rj) node [above] {\small {$R_m(Z_m)\rightarrow v$}};
\draw (wj) node [above] {\small {$W_m(X_m,nv)$}};
\draw (cj) node [above] {\small {$\TryC_m$}};


\node (p1) at (18,0) [] {};
\node (pj) at (22,0) [] {};

\draw (p1) node [above] {\small {$R_{n}(X_1)\rightarrow v$}};
\draw (pj) node [above] {\small {$R_{n}(X_j)\rightarrow v$}};
\node (z1) at (15,0) [] {};

\draw (z1) node [above] {\small {$W_{n-2}(Z_j,nv)$}};

\begin{scope}   
\draw [|-,thick] (0,0) node[left] {$T_1$} to (5,0);
\draw [|-|,thick] (0,0) node[left] {} to (2,0);
\draw [-|,thick] (2,0) node[left] {} to (4,0);
\draw [-,dotted] (5,0) node[left] {} to (6,0);
\end{scope}
\begin{scope}   
\draw [|-|,thick] (7,0) node[left] {$T_m$} to (9,0);
\draw [-|,thick] (9,0) node[left] {} to (11,0);
\draw [-,thick] (11,0) node[left] {} to (12,0);
\draw [-,dotted] (11,0) node[left] {} to (13,0);
\end{scope}
\begin{scope}   
\draw [|-,dotted] (17,0) node[left] {$T_{n}$} to (23.5,0);
\draw [|-|,thick] (17,0) to (19,0);
\draw [|-|,thick] (21,0) to (23,0);
\end{scope}
\begin{scope}   
\draw [|-|,thick] (14,0) node[left] {$T_{n-2}$} to (16,0);
\end{scope}
%

\end{tikzpicture}}}
        \\
        \vspace{2mm}
        \subfloat[Suppose $\Read_{n}(X_j)$ does not perform a RAW/AWAR, 
        $T_{n}$ and $T_{n-1}$ are unaware of step contention and $T_n$ misses the event of $T_j$, but $R_{n-1}(X_j)$ returns the value of $X_j$ 
        that is updated by $T_j$\label{sfig:rw-4}]{\scalebox{0.5}[0.5]{\begin{tikzpicture}
\node (r1) at (1,0) [] {};
\node (w1) at (3,0) [] {};
\node (c1) at (5,0) [] {};

\draw (r1) node [above] {\small {$R_1(Z_1)\rightarrow v$}};
\draw (w1) node [above] {\small {$W_1(X_1,nv)$}};
\draw (c1) node [above] {\small {$\TryC_1$}};


\node (rj) at (8,0) [] {};
\node (wj) at (10.2,0) [] {};
\node (cj) at (12,0) [] {};

\draw (rj) node [above] {\small {$R_m(Z_m)\rightarrow v$}};
\draw (wj) node [above] {\small {$W_m(X_m,nv)$}};
\draw (cj) node [above] {\small {$\TryC_m$}};


\node (p1) at (18,0) [] {};
\node (pj) at (25,0) [] {};

\draw (p1) node [above] {\small {$R_{n}(X_1)\rightarrow v$}};
\draw (pj) node [above] {\large {$R_{n}(X_j)\rightarrow v$}};
\node (z1) at (15,0) [] {};

\draw (z1) node [above] {\small {$W_{n-2}(Z_j,nv)$}};
\node (e) at (21.1,-1) [] {};
\node (l1) at (24,-1) [] {};
\node (lj) at (28,-1) [] {};

\draw (e) node [above] {\large {(event of $T_j$)}};
\draw (l1) node [above] {\large {$R_{n-1}(X_1)$}};
\draw (lj) node [above] {\large {$R_{n-1}(X_j)\rightarrow nv$}};

\begin{scope}   
\draw [|-,thick] (0,0) node[left] {$T_1$} to (5,0);
\draw [|-|,thick] (0,0) node[left] {} to (2,0);
\draw [-|,thick] (2,0) node[left] {} to (4,0);
\draw [-,dotted] (5,0) node[left] {} to (6,0);
\end{scope}
\begin{scope}   
\draw [|-|,thick] (7,0) node[left] {$T_m$} to (9,0);
\draw [-|,thick] (9,0) node[left] {} to (11,0);
\draw [-,thick] (11,0) node[left] {} to (12,0);
\draw [-,dotted] (11,0) node[left] {} to (13,0);
\end{scope}
\begin{scope}   
\draw [|-,dotted] (17,0) node[left] {$T_{n}$} to (30,0);
\draw [|-|,thick] (17,0) to (19,0);
\draw [|-|,thick] (21,0) to (30,0);
\end{scope}
\begin{scope}   
\draw [|-|,thick] (14,0) node[left] {$T_{n-2}$} to (16,0);
\end{scope}
\begin{scope}
\draw  (21.2,-1) circle [fill, radius=0.05]  (21.2,-1);
\draw [-,dotted] (23.2,-1) node[left] {\large $T_{n-1}$} to (29,-1);
\draw [|-|,thick] (23.2,-1) node[left] {} to (25.1,-1);
\draw [|-|,thick] (27,-1) node[left] {} to (29,-1);
\end{scope}

\end{tikzpicture}}}
                
	\caption{Executions in the proof of Theorem~\ref{th:oftriv}; execution in \ref{sfig:rw-4} is not opaque
          \label{fig:rw}} 
\end{center}
\end{figure*}
\begin{theorem}
\label{th:oftriv}
Every RW DAP opaque TM implementation $M\in \mathcal{OF}$ has
an execution $E$ in which some read-only transaction $T \in \ms{txns}(E)$ performs $\Omega(n)$ non-overlapping RAW/AWARs.
\end{theorem}
\begin{proof}
For all $j\in \{1,\ldots, m\}$; $m=n-3$, let $v$ be the initial value of t-objects $X_j$ and $Z_j$.
Throughout this proof, we assume that, for all $i\in \{1,\ldots , n\}$, transaction $T_i$ is invoked by process $p_i$.

By OF TM-progress and OF TM-liveness, any opaque and RW DAP TM implementation $M\in \mathcal{OF}$ 
has an execution of the form ${\bar \rho}_1\cdots {\bar \rho}_m$, where for all $j\in \{1,\ldots, m\}$, ${\bar \rho}_j$ 
denotes the t-complete step contention-free execution of transaction $T_j$
that performs $\Read_j(Z_j)\rightarrow v$, writes value $nv \neq v$ to 
$X_j$ and commits.

By construction, any two transactions that participate in ${\bar \rho}_1\cdots {\bar \rho}_n$
are mutually read-write disjoint-access and cannot contend on the same base object.
It follows that for all $1\leq j\leq m$,
${\bar \rho}_j$ is an execution of $M$.

For all $j\in \{1,\ldots, m\}$, we iteratively define an execution $\rho_j$ of $M$ as follows:  
it is the longest prefix of ${\bar \rho}_j$ such that 
$\rho_1\cdots \rho_{j}$ cannot be extended with the complete step contention-free execution fragment
of transaction $T_{n}$ that performs $j$ t-reads: $\Read_n(X_1)\cdots \Read_n(X_j)$
in which $\Read_{n}(X_j) \rightarrow nv$ nor with the 
complete step contention-free execution fragment
of transaction $T_{n-1}$ that performs $j$ t-reads: $\Read_{n-1}(X_1)\cdots \Read_{n-1}(X_j)$ 
in which $\Read_{n-1}(X_j) \rightarrow nv$ (Figure~\ref{sfig:rw-1}).

For any $j\in \{1,\ldots, m\}$, let $e_j$ be the event
transaction $T_j$ is poised to apply in the configuration after $\rho_1\cdots \rho_j$.
Thus, the execution $\rho_1\cdots \rho_j\cdot e_j$ can be extended with the complete step contention-free
executions of at least one of transaction $T_{n}$ or $T_{n-1}$ that performs $j$ t-reads of $X_1, \ldots ,X_j$
in which the t-read of $X_j$ returns the new value $nv$. 
Let $T_{n-1}$ be the transaction that must return the new value for the maximum number of $X_j$'s when
$\rho_1\cdots \rho_j\cdot e_j$ is extended with the t-reads of $X_1,\ldots , X_j$.
We show that, in the worst-case, transaction $T_n$ must perform $\lceil \frac{m}{2} \rceil$ non-overlapping RAW/AWARs
in the course of performing $m$ t-reads of $X_1,\ldots , X_m$ immediately after $\rho_1\cdots \rho_m$.
Symmetric arguments apply for the case when $T_{n}$
must return the new value for the maximum number of $X_j$'s when
$\rho_1\cdots \rho_j\cdot e_j$ is extended with the t-reads of $X_1,\ldots , X_j$.

\vspace{1mm}\noindent\textbf{Proving the RAW/AWAR lower bound.}
We prove that transaction $T_n$ must perform
$\lceil \frac{m}{2} \rceil$ non-overlapping RAWs or AWARs
in the course of performing $m$ t-reads of $X_1,\ldots , X_m$ immediately after the execution
$\rho_1\cdots \rho_m$.
Specifically, we prove that $T_n$ must perform a RAW or an AWAR during the execution of the t-read of each $X_j$
such that $\rho_1\cdots \rho_j\cdot e_j$ can be extended with the complete step contention-free
execution of $T_{n-1}$ as it performs $j$ t-reads of $X_1\ldots X_j$
in which the t-read of $X_j$ returns the new value $nv$.
Let $\mathbb{J}$ denote the of all $j\in \{1,\ldots , m\}$ such that
$\rho_1\cdots \rho_j\cdot e_j$ extended with the complete step contention-free
execution of $T_{n-1}$ performing $j$ t-reads of $X_1\ldots X_j$ must
return the new value $nv$ during the t-read of $X_j$.

We first prove that, for all $j \in \mathbb{J}$, 
$M$ has an execution of the form $\rho_1\cdots \rho_m \cdot \delta_j$ 
(Figures~\ref{sfig:rw-1} and \ref{sfig:rw-2}),
where $\delta_j$ is the complete step contention-free
execution fragment of $T_{n}$ that performs $j$ t-reads: $\Read_{n}(X_1)  \cdots \Read_{n}(X_j)$, each of which
return the initial value $v$.

By definition of $\rho_j$, OF TM-progress and OF TM-liveness, 
$M$ has an execution of the form $\rho_1\cdots \rho_j \cdot \delta_j$. 
By construction, transaction $T_{n}$ is read-write disjoint-access with each transaction $T \in \{T_{j+1},\ldots , T_m\}$ 
in $\rho_1\cdots \rho_j \cdots \rho_m \cdot \delta_j$. Thus, $T_n$ cannot contend with any of the transactions
in $\{T_{j+1},\ldots , T_m\}$, implying that, for all $j\in \{1,\ldots , m\}$, $M$ has an execution of the form
$\rho_1\cdots \rho_m \cdot \delta_j$ (Figure~\ref{sfig:rw-2}).

We claim that, for each $j\in \mathbb{J}$, the t-read of $X_j$ performed by $T_n$ must perform a RAW or an AWAR in the course
of performing $j$ t-reads of $X_1, \ldots , X_j$ immediately after $\rho_1\cdots \rho_m$.
Suppose by contradiction that $\Read_n(X_j)$ does not perform a RAW or an AWAR in 
$\rho_1\cdots \rho_m \cdot \delta_m$.
\begin{claim}
\label{cl:oftmex0}
For all $j \in \mathbb{J}$, $M$ has an execution of the form
$\rho_1\cdots \rho_{j} \cdots \rho_m  \cdot  \delta_{j-1} \cdot e_j \cdot \beta$ where,
$\beta$ is the complete step contention-free execution fragment of transaction $T_{n-1}$ that
performs $j$ t-reads: $\Read_{n-1}(X_1)\cdots \Read_{n-1}(X_{j-1}) \cdot \Read_{n-1}(X_j)$
in which $\Read_{n-1}(X_j)$ returns $nv$.
\end{claim}
\begin{proof}
We observe that transaction $T_{n}$ is read-write disjoint-access with every transaction $T\in \{T_j,T_{j+1},\ldots , T_m\}$ in
$\rho_1\cdots \rho_{j}\cdots \rho_m \cdot  \delta_{j-1}$.
By RW DAP, it follows that $M$ has an execution of the form
$\rho_1\cdots \rho_{j} \cdots \rho_m  \cdot  \delta_{j-1} \cdot e_j$
since $T_n$ cannot perform a nontrivial event on the base object accessed by $T_j$ in the event $e_j$.

By the definition of $\rho_j$, 
transaction $T_{n-1}$ must access the base object to which
$T_j$ applies a nontrivial primitive in $e_j$ to return the value $nv$ of $X_j$ as it performs $j$ t-reads
of $X_1,\ldots , X_j$ immediately after the execution $\rho_1\cdots \rho_{j} \cdots \rho_m  \cdot  \delta_{j-1} \cdot e_j$.
Thus, $M$ has an execution of the form $\rho_1\cdots \rho_{j}  \cdot  \delta_{j-1} \cdot e_j \cdot \beta$.

By construction, transactions $T_{n-1}$ is read-write disjoint-access with every
transaction $T\in \{T_{j+1},\ldots , T_m\}$ in
$\rho_1\cdots \rho_{j} \cdots \rho_m  \cdot  \delta_{j-1} \cdot e_j \cdot \beta$.
It follows that $M$ has an execution of the form 
$\rho_1\cdots \rho_j \cdots \rho_m \cdot \delta_{j-1} \cdot e_j \cdot \beta$.
\end{proof}
\begin{claim}
\label{cl:oftmex1}
For all $j \in \{1,\ldots , m\}$, 
$M$ has an execution of the form 
$\rho_1\cdots \rho_{j} \cdots \rho_m  \cdot \gamma\cdot  \delta_{j-1} \cdot e_j \cdot \beta$,
where $\gamma$ is the t-complete step contention-free execution fragment
of transaction $T_{n-2}$ that writes $nv \neq v$ to $Z_j$ and commits.
\end{claim}
\begin{proof}
Observe that $T_{n-2}$ precedes transactions $T_{n}$ and $T_{n-1}$ in real-time order in the above execution.

By OF TM-progress and OF TM-liveness, transaction $T_{n-2}$ must be committed in
$\rho_1\cdots \rho_{j} \cdots \rho_m  \cdot \gamma$.

Since transaction $T_{n-1}$ is read-write disjoint-access
with $T_{n-2}$ in $\rho_1\cdots \rho_{j} \cdots \rho_m  \cdot \gamma \cdot  \delta_{j-1} \cdot e_j \cdot \beta$,
$T_{n-1}$ does not contend with $T_{n-2}$ on any base object (recall that we associate an edge
with t-objects in the conflict graph only if they are both contained in the write set of some transaction).
Since the execution fragment $\beta$ contains an access to the base object to which $T_j$ performs
a nontrivial primitive in the event $e_j$, $T_{n-2}$ cannot perform a nontrivial event on this base object
in $\gamma$.
It follows that $M$ has an execution of the form 
$\rho_1\cdots \rho_{j} \cdots \rho_m  \cdot \gamma\cdot  \delta_{j-1} \cdot e_j \cdot \beta$
since, it is indistinguishable to $T_{n-1}$
from the execution 
$\rho_1\cdots \rho_{j} \cdots \rho_m  \cdot  \delta_{j-1} \cdot e_j \cdot \beta$ (the existence of which is already
established in Claim~\ref{cl:oftmex0}).
\end{proof}
Recall that transaction $T_n$ is read-write disjoint-access with $T_{n-2}$ in 
$\rho_1\cdots \rho_{j} \cdots \rho_m  \cdot \gamma \cdot  \delta_{j}$.
Thus, $M$ has an execution of the form 
$\rho_1\cdots \rho_{j} \cdots \rho_m  \cdot \gamma\cdot  \delta_{j}$ (Figure~\ref{sfig:rw-3}).

\vspace{1mm}\noindent\textbf{Deriving a contradiction.}
For all $j \in \{1,\ldots , m\}$, we represent the execution fragment 
$\delta_j$ as $\delta_{j-1} \cdot \pi^j$, where $\pi^j$ is the complete execution
fragment of the $j^{th}$ t-read $\Read_{n}(X_j) \rightarrow v$.
By our assumption, $\pi^j$ does not contain a RAW or an AWAR.

For succinctness, let $\alpha=\rho_1\cdots \rho_m \cdot \gamma \cdot \delta_{j-1}$.
We now prove that if $\pi^j$ does not contain a RAW or an AWAR, we can define $\pi^j_{1} \cdot \pi^j_{2}=\pi^j$
to construct an execution of the form
$\alpha  \cdot \pi^j_{1} \cdot e_j \cdot \beta \cdot \pi^j_{2}$ (Figure~\ref{sfig:rw-4}) such that
\begin{itemize}
\item
no event in $\pi^j_1$ is the application of a nontrivial primitive
\item 
$\alpha \cdot \pi^j_{1} \cdot e_j \cdot \beta \cdot \pi^j_{2}$
is indistinguishable to $T_{n}$ from the step contention-free execution
$\alpha \cdot \pi^j_{1} \cdot \pi^j_{2}$
\item
$\alpha \cdot \pi^j_{1}  \cdot e_j \cdot \beta \cdot \pi^j_{2}$
is indistinguishable to $T_{n-1}$ from the step contention-free execution
$\alpha \cdot e_j \cdot \beta$.
\end{itemize}
The following claim defines $\pi^j_1$ and $\pi^j_2$ to construct this execution.
\begin{claim}
\label{cl:ofraw}
For all $j \in \{1,\ldots , m\}$,
$M$ has an execution of the form
$\alpha  \cdot \pi^j_{1} \cdot e_j \cdot \beta \cdot \pi^j_{2}$.
\end{claim}
\begin{proof}
Let $t$ be the first event containing a write to a base object in the execution fragment $\pi^j$.
We represent $\pi^j$ as the execution fragment $\pi^j_{1}\cdot t \cdot \pi^j_{f}$.
Since $\pi^j_1$ does not contain nontrivial events that write to a base object, $\alpha \cdot \pi^j_{1}  \cdot e_j\cdot \beta$
is indistinguishable to transaction $T_{n-1}$ from the step contention-free execution
$\alpha  \cdot e_j\cdot \beta$ (as already proven in Claim~\ref{cl:oftmex1}).
Consequently, $\alpha \cdot \pi^j_{1} \cdot e_j\cdot \beta$
is an execution of $M$.

Since $t$ is not an atomic-write-after-read, $M$ has an execution of the form
$\alpha\cdot \gamma\cdot  \pi^j_{1} \cdot e_j\cdot \beta \cdot t$.
Secondly, since $\pi^j$ does not contain a read-after-write, any read of a base object performed in $\pi^j_{f}$ 
may only be performed to base objects previously written in $t \cdot \pi^j_{f}$.
Thus, $\alpha \cdot \pi^j_{1} \cdot e_j\cdot \beta \cdot t \cdot \pi^j_{f}$
is indistinguishable to $T_{n}$ from the step contention-free execution
$\alpha \cdot \pi^j_1 \cdot t \cdot \pi^j_{f}$. 
But, as already proved, $\alpha \cdot \pi^{j}$ is an execution of $M$.

Choosing $\pi^j_{2}=t \cdot \pi^j_{f}$, it follows that
$M$ has an execution of the form $\alpha  \cdot \pi^j_{1} \cdot  e_j\cdot \beta \cdot \pi^j_{2}$.
\end{proof}
We have now proved that, for all $j \in \{1,\ldots , m\}$,
$M$ has an execution of the form
$\rho_1\cdots \rho_m \cdot \gamma \cdot \delta_{j-1} \cdot \pi^j_{1} \cdot  e_j\cdot \beta \cdot \pi^j_{2}$ (Figure~\ref{sfig:rw-4}). 

The execution in Figure~\ref{sfig:rw-4}
is not opaque.
Indeed, in any serialization the following must hold. 
Since $T_{n-1}$ reads the value written by $T_j$ in $X_j$,
$T_j$ must be committed.  
Since $\Read_{n}(X_j)$ returns the initial value $v$, 
$T_{n}$ must precede $T_j$.
The committed transaction $T_{n-2}$, which writes a new value to $Z_j$, must precede 
$T_{n}$ to respect the real-time order on transactions. 
However, $T_j$ must precede $T_{n-2}$ since $\Read_j(Z_j)$ returns the initial value of $Z_j$.
The cycle $T_j\rightarrow T_{n-2} \rightarrow T_{n} \rightarrow T_j$
implies that there exists no such a serialization.

Thus, for each $j\in \mathbb{J}$, transaction $T_n$ must perform a RAW or an AWAR during 
the t-read of $X_j$ in the course
of performing $m$ t-reads of $X_1, \ldots , X_m$ immediately after $\rho_1\cdots \rho_m$.
Since $|\mathbb{J}|\geq \lceil \frac{(n-3)}{2} \rceil$, in the worst-case, $T_n$ must perform $\Omega(n)$ RAW/AWARs during the execution of $m$ t-reads
immediately after $\rho_1\cdots \rho_m$.
\end{proof}
%
%
%
\section{Upper bound for opaque progressive TMs}
\label{sec:iclb}
In this section, we describe a progressive, opaque TM implementation $LP$ (Algorithm~\ref{alg:ic})
that is not subject to any of the lower bounds inherent to implementations in $\mathcal{OF}$ (cf. Figure~\ref{fig:main}).
Our implementation satisfies strict DAP, every transaction performs at most a single RAW and every t-read operation
incurs $O(1)$ memory stalls in any execution.

\vspace{1mm}\noindent\textbf{Base objects.}
For every t-object $X_j$, $LP$ maintains a base object $v_j$ that stores the \emph{value} of $X_j$.
Additionally, for each $X_j$, there is a \emph{bit} $L_j$, which if set, indicates the presence of an updating transaction
writing to $X_j$.
For every process $p_i$ and t-object $X_j$,
$LP$ maintains a \emph{single-writer bit} $r_{ij}$ (only $p_i$ is allowed to write to $r_{ij}$).
Each of these base objects may be accessed only via read and write primitives.

\vspace{1mm}\noindent\textbf{Updating transactions.}
The $\Write_k(X,v)$ implementation by process $p_i$ simply stores the value $v$ locally, deferring the actual updates
to $\TryC_k$.
During $\TryC_k$, process $p_i$ attempts to obtain exclusive write access to every
$X_j\in \Wset(T_k)$. This is realized through the single-writer bits, which ensure that no other transaction
may write to base objects $v_j$ and $L_j$ until $T_k$ relinquishes its exclusive write access to $\Wset(T_k)$.
Specifically, process $p_i$ writes $1$ to each $r_{ij}$, then checks
that no other process $p_t$ has written $1$ to any $r_{tj}$ by executing a series of reads (incurring a single RAW).
If there exists such a process that concurrently contends on write set of $T_k$, 
for each $X_j\in \Wset(T_k)$, $p_i$ writes $0$ to $r_{ij}$ and returns $A_k$. 
If successful in obtaining exclusive write access to $\Wset(T_k)$, $p_i$ sets
the bit $L_j$ for each $X_j$ in its write set.
Implementation of $\TryC_k$ now checks if any t-object in its read set is concurrently contended by another transaction
and then validates its read set. 
If there is contention on the read set or validation fails, indicating the presence of a concurrent conflicting
transaction, the transaction is aborted. If not, $p_i$ writes the values of the t-objects to shared memory and 
relinquishes exclusive write access to each $X_j \in \Wset(T_k)$ by writing $0$ to each of the base objects $L_j$ 
and $r_{ij}$.

\vspace{1mm}\noindent\textbf{Read operations.}
The implementation first reads the value of t-object $X_j$ from base object $v_j$ and then
reads the bit $L_j$ to detect contention with an updating transaction.
If $L_j$ is set, the transaction is aborted; if not, read validation is performed on the entire read set. If the validation fails,
the transaction is aborted. Otherwise, the implementation returns the value of $X_j$.
For a read-only transaction $T_k$, $\TryC_k$ simply returns the commit response.

\vspace{1mm}\noindent\textbf{Complexity.}
Observe that our implementation uses invisible reads since read-only transactions do not apply any nontrivial primitives.
Any updating transaction performs at most a single RAW in the course of acquiring exclusive write access to
the transaction's write set. 
Consequently, every transaction performs $O(1)$ non-overlapping RAWs in any execution. 

Recall that a transaction may write to base objects $v_j$ and $L_j$ only after obtaining exclusive write access to t-object $X_j$, 
which in turn is realized via single-writer base objects.
Thus, no transaction performs a write to any base object $b$ 
immediately after a write to $b$ by 
another transaction, \emph{i.e.}, every transaction incurs only $O(1)$ memory stalls on account of any event it performs.
Since the $\Read_k(X_j)$ implementation only accesses base objects $v_j$ and $L_j$, and the validating $T_k$'s read set does not
cause any stalls, it follows that each t-operation performs $O(1)$ stalls in every execution.

Moreover, $LP$ ensures that any two transactions $T_i$ and $T_j$ access the same base object \emph{iff}
there exists $X\in \Dset(T_i) \cap \Dset(T_j)$ (strict DAP)
and maintains exactly one version for every t-object at any prefix of the execution.
\begin{theorem}
\label{th:ic}
Algorithm~\ref{alg:ic} describes a progressive, opaque and strict DAP TM implementation $LP$ that provides
wait-free TM-liveness, uses invisible reads and in every execution $E$ of $LP$, 
\begin{itemize}
\item every transaction $T\in \ms{txns}(E)$ applies only read and write primitives in $E$,
\item every transaction $T \in \ms{txns}(E)$ performs at most a single RAW,
\item for every transaction $T \in \ms{txns}(E)$, every t-read operation performed by $T$
incurs $O(1)$ memory stalls in $E$.
\end{itemize}
\end{theorem}
%
%
%
%
\section{Related work}
\label{sec:related}
The lower bounds and impossibility results presented in this paper
apply to obstruction-free TMs, such as DSTM~\cite{HLM+03},
FSTM~\cite{fraser}, and others~\cite{astm,nztm,fraser}. 
Our upper bound is inspired by the progressive TM of~\cite{KR11}. 

%
Attiya \emph{et al.}~\cite{AHM09} were the first to formally define DAP for TMs.
They proved the impossibility of implementing weak DAP strictly serializable
TMs that use invisible reads and guarantee that read-only transactions eventually commit, while 
updating transactions are guaranteed to commit only when they run sequentially~\cite{AHM09}.
This class is orthogonal to the class of
obstruction-free TMs, as is the proof technique used to establish the impossibility.

Perelman \emph{et al.}~\cite{PFK10} showed that
\emph{mv-permissive} weak DAP TMs cannot be implemented. 
In mv-permissive TMs, only updating transactions may
be aborted, and only when they conflict with other updating transactions.
In particular, read-only transactions cannot be aborted and updating
transactions may sometimes be aborted even in the absence of step contention, which makes
the impossibility result in~\cite{PFK10} unrelated to ours.  

Guerraoui and Kapalka~\cite{tm-book} proved that it is impossible to implement
strict DAP obstruction-free TMs. They also proved that a strict serializable TM that provides OF TM-progress and
wait-free TM-liveness cannot be implemented using only read and write primitives.
We show that progressive TMs are not subject to either of these lower bounds.

Attiya \emph{et al.} introduced the RAW/AWAR metric and proved that it is impossible to derive RAW/AWAR-free implementations of
a wide class of data types that include \emph{sets}, \emph{queues} and \emph{deadlock-free mutual exclusion}.
The metric was previously used in~\cite{KR11}
to measure the complexity of read-only transactions in a strictly
stronger (than $\mathcal{OF}$) class of \emph{permissive} TMs.
Detailed coverage on memory fences and the RAW/AWAR metric can be found in \cite{McKenney10}.

To derive the linear lower bound on the memory stall complexity of obstruction-free TMs,
we adopted the definition of a $k$-\emph{stall execution} and certain proof steps 
from~\cite{G05,AGHK09}.
%

%
\section{Discussion}
\label{sec:disc}
\vspace{1mm}\noindent\textbf{Lower bounds for obstruction-free TMs.}
We chose obstruction-freedom to elucidate non-blocking TM-progress since it is a very weak
non-blocking progress condition~\cite{HS11-progress}. 
As highlighted in the paper by Ennals~\cite{Ennals05}, (1) obstruction-freedom
increases the number of concurrently executing transactions since transactions cannot wait for inactive transactions to complete,
and (2) while performing a t-read, obstruction-free TMs like \cite{HLM+03, fraser} must forcefully
abort pending conflicting transactions.
Intuitively, (1) allows us to construct executions in which
some pending transaction is stalled while accessing a base object by all other concurrent transactions waiting to apply 
nontrivial primitives on the base object.
Observation (2) inspires the proof of the impossibility
of invisible reads in Theorem~\ref{th:ir}. Typically, the reading transaction must acquire exclusive ownership
of the object via mutual exclusion or employing a read-modify-write primitive like \emph{compare-and-swap},
motivating the linear lower bound on expensive synchronization in Theorem~\ref{th:oftriv}.
In practice though, obstruction-free TMs may possibly circumvent these lower bounds in models that allow the 
use of \emph{contention managers}~\cite{dstm-contention}.

Observe that Theorems~\ref{th:ir} and \ref{th:oftmstalls} assume strict serializability and thus,
also hold under the assumption of stronger TM-correctness conditions like opacity, 
\emph{virtual-world consistency}~\cite{damien-vw} and \emph{TMS}~\cite{TMS}.

Since there are at most $n$ concurrent transactions,
we cannot do better than $(n-1)$ stalls (cf. Definition~\ref{def:stalls}).
Thus, the lower bound of Theorem~\ref{th:oftmstalls} is tight. 
Moreover, we conjecture that the linear (in $n$) lower bound of Theorem~\ref{th:oftriv} for RW DAP opaque obstruction-free TMs 
can be strengthened to be linear in the size of the transaction's read set.  
Then, Algorithm~\ref{alg:oftm} in Appendix~\ref{app:upperof}
would allow us to
establish a linear tight bound (in the size of the transaction's read set) for RW DAP opaque obstruction-free TMs.

\vspace{1mm}\noindent\textbf{Progressive vs. obstruction-free TMs.}
Progressiveness is a blocking TM-progress condition that is satisfied by several popular TM implementations
like \emph{TL2}~\cite{DSS06} and \emph{NOrec}~\cite{norec}.
In general, progressiveness and obstruction-freedom are incomparable.
On the one hand, 
a t-read $X$ by a transaction $T$ that runs step contention-free
from a configuration that contains an incomplete t-write to $X$ is typically \emph{blocked} or aborted in lock-based TMs; 
obstruction-free TMs however, must ensure that $T$ must complete its t-read of $X$ without blocking or aborting.
On the other hand, progressiveness requires two non-conflicting transactions to commit even
in executions that are not step contention-free; but this is not guaranteed by obstruction-freedom.

Intuitively, progressive implementations are not forced to abort
conflicting transactions, which allows us to employ invisible reads, derive constant stall and RAW/AWAR implementations.
While it is relatively easy to derive standalone progressive TM implementations that are not individually 
subject to the lower bounds
of obstruction-free TMs (cf. Figure~\ref{fig:main}), our progressive opaque TM implementation $LP$ 
is not subject to any of the lower bounds we prove for implementations in $\mathcal{OF}$.

Circa. 2005, several papers presented the case for a shift from TMs that provide obstruction-free TM-progress
to lock-based progressive TMs~\cite{DStransaction06,DSS06, Ennals05}. 
They argued that lock-based TMs tend to outperform
obstruction-free ones by allowing for simpler algorithms with lower overheads and their inherent progress issues
may be resolved using timeouts and contention-managers.
The lower bounds for non-blocking TMs and the complexity gap with our progressive TM implementation established
in this paper suggest that this course correction was indeed justified.
\newpage
\bibliography{references}
\appendix
\section{Opaque progressive TM implementation $LP$}
\label{app:upper}
\begin{algorithm}[!h]
\caption{Strict DAP progressive opaque TM implementation $LP$; code for $T_k$ executed by process $p_i$}
\label{alg:ic}
\begin{algorithmic}[1]
  	\begin{multicols}{2}
  	{\footnotesize
	\Part{Shared base objects}{
		\State $v_j$, for each t-object $X_j$, allows reads and writes
		\State $r_{ij}$, for each process $p_i$ and t-object $X_j$
		\State ~~~~~single-writer bit
		\State ~~~~~allows reads and writes
		\State $L_j$, for each t-object $X_j$ 
		\State ~~~~~allows reads and writes
	}\EndPart
	\Part{Local variables}{
		\State $\ms{Rset}_k,\ms{Wset}_k$ for every transaction $T_k$;
		\State ~~~~dictionaries storing $\{X_m, v_m\}$
	}\EndPart	

	\Statex
	\Part{\Read$_k(X_j)$}{
		\If{$X_j \not\in \Rset(T_k)$}
		
		\State $[\textit{ov}_j,k_j ] := \Read(v_j)$ \label{line:read2}
		\State $\Rset(T_k) := \Rset(T_k)\cup\{X_j,[\textit{ov}_j,k_j]\}$ \label{line:rset}
		\If{$\Read(L_j)\neq 0$} \label{line:abort0}
			\Return $A_k$ \EndReturn
		\EndIf
		\If{$\lit{validate}()$} \label{line:read-validate}
			\Return $A_k$ \EndReturn
		\EndIf
		\Return $\textit{ov}_j$ \EndReturn
		
		\Else
		    
		\State $[\textit{ov}_j, \bot] :=\Rset(T_k).\lit{locate}(X_j)$
		\Return $\textit{ov}_j$ \EndReturn
		
		\EndIf
   	}\EndPart
	\Statex
	\Part{\Write$_k(X_j,v)$}{
		\State $\textit{nv}_j := v$
		\State $\Wset(T_k) := \Wset(T_k)\cup\{X_j\}$
		\Return $\ok$ \EndReturn
		
   	}\EndPart
	\Statex
%
	
	\Part{\TryC$_k$()}{
		\If{$|\Wset(T_k)|= \emptyset$}
			\Return $C_k$ \EndReturn \label{line:return}
		\EndIf
				
		\State locked := $\lit{acquire}(\Wset(T_k))$\label{line:acq} 
		\If{$\neg$ locked} \label{line:abort2} 
	 		\Return $A_k$ \EndReturn
	 	\EndIf
	 	
		\If{$\lit{isAbortable}()$} \label{line:abort3}
			\State $\lit{release}(\Wset(T_k))$ 
			\Return $A_k$ \EndReturn
		\EndIf
		\Statex
		\Comment{Exclusive write access to each $v_j$}
		\ForAll{$X_j \in \Wset(T_k)$}
	 		 \State $\Write(v_j,[\textit{nv}_j,k])$ \label{line:write}
	 	\EndFor		
		\State $\lit{release}(\Wset(T_k))$   	\label{line:rellock}	
   		\Return $C_k$ \EndReturn
   	 }\EndPart		
	
	\newpage
	\Part{Function: $\lit{release}(Q)$}{
		\ForAll{$X_j \in Q$}	
 			\State \Write$(L_{j},0)$ \label{line:wlockrelease}
		\EndFor
  		\ForAll{$X_j \in Q$}	
 			\State \Write$(r_{ij},0)$ \label{line:rel1}
		\EndFor
		\Return $ok$ \EndReturn
	}\EndPart

 	\Statex
 	\Part{Function: $\lit{acquire}(Q$)}{
   		\ForAll{$X_j \in Q$}	
			\State \Write$(r_{ij},1)$ \label{line:acq1}
		\EndFor
		\If{$\exists X_j \in Q;t\neq k : \Read(r_{tj})=1$} \label{line:lock}
			\ForAll{$X_j \in Q$}	
				\State \Write$(r_{ij},0)$
			\EndFor
			\Return $\false$ \EndReturn
		\EndIf
		\Statex
		\Comment{Exclusive write access to each $L_j$}
		\ForAll{$X_j \in Q$}
		  \State $\Write(L_j,1)$ \label{line:wlockwrite}
		\EndFor
		\Return $\true$ \EndReturn
	}\EndPart		
	 \Statex
	 
	\Part{Function: $\lit{isAbortable()}$ }{
		\If{$\exists X_j \in \Rset(T_k): X_j\not\in \Wset(T_k)\wedge \Read(L_j)\neq 0$} \label{line:valid0} \label{line:isl}
			\Return $\true$ \EndReturn
		\EndIf
		
		\If{$\lit{validate}()$} 
			\Return $\true$ \EndReturn
		\EndIf
		\Return $\false$ \EndReturn
	}\EndPart
	\Statex
	\Part{Function: $\lit{validate()}$ }{
		
		\Comment{Read validation}
		\If{$\exists X_j \in Rset(T_k)$:$[\textit{ov}_j,k_j]\neq \Read(v_j)$} \label{line:valid} 
			\Return $\true$ \EndReturn
		\EndIf
		\Return $\false$ \EndReturn
	}\EndPart
	}
	\end{multicols}
  \end{algorithmic}
\end{algorithm}
In this section, we describe our blocking TM implementation $LP$ that satisfies progressiveness and opacity~\cite{tm-book}.
We begin with the formal definition of \emph{opacity}.

For simplicity of presentation, we assume that each execution $E$
begins with an ``imaginary'' transaction $T_0$ that writes initial
values to all t-objects and commits before any other transaction
begins in $E$.
Let $E$ be a 
t-sequential execution.
For every operation $\Read_k(X)$ in $E$,
we define the \emph{latest written value} of $X$ as follows:
(1) If $T_k$ contains a $\Write_k(X,v)$ preceding $\Read_k(X)$,
then the latest written value of $X$ is the value of the latest such write to $X$.
(2) Otherwise, if $E$ contains a $\Write_m(X,v)$,
$T_m$ precedes $T_k$, and $T_m$ commits in $E$,
then the latest written value of $X$ is the value
of the latest such write to $X$ in $E$.
(This write is well-defined since $E$ starts with $T_0$ writing to
all t-objects.)
We say that $\Read_k(X)$ is \emph{legal} in a t-sequential execution $E$ if it returns the
latest written value of $X$, and $E$ is \emph{legal}
if every $\Read_k(X)$ in $H$ that does not return $A_k$ is legal in $E$.

For a history $H$, a \emph{completion of $H$}, denoted ${\bar H}$,
is a history derived from $H$ through the following procedure:
(1) for every incomplete t-operation $op_k$ of $T_k \in \txns(H)$ in $H$,
if $op_k=\Read_k \vee \Write_k$, 
insert $A_k$ somewhere after the invocation of $op_k$; 
otherwise, if $op_k=\TryC_k()$, 
insert $C_k$ or $A_k$ somewhere after the last event of $T_k$.
(2) for every complete transaction $T_k$ that is not t-complete, insert $\textit{tryC}_k\cdot A_k$ somewhere after the 
last event of transaction $T_k$.
\begin{definition}
\label{def:opaque}
A finite history $H$ is \emph{opaque} if there
is a legal t-complete t-sequential history $S$,
such that
(1) for any two transactions $T_k,T_m \in \txns(H)$,
if $T_k \prec_H^{RT} T_m$, then $T_k$ precedes $T_m$ in $S$, and
(2) $S$ is equivalent to a completion of $H$.

A finite history $H$ is \emph{strictly serializable} if there
is a legal t-complete t-sequential history $S$,
such that
(1) for any two transactions $T_k,T_m \in \txns(H)$,
if $T_k \prec_H^{RT} T_m$, then $T_k$ precedes $T_m$ in $S$, and
(2) $S$ is equivalent to $\ms{cseq}(\bar H)$, where $\bar H$ is some
completion of $H$ and $\ms{cseq}(\bar H)$ is the subsequence of $\bar H$ reduced to
committed transactions in $\bar H$.

We refer to $S$ as a \emph{serialization} of $H$.
\end{definition}
We now prove that $LP$ implements an opaque TM.

We introduce the following technical definition:
process $p_i$ \emph{holds a lock on $X_j$ after an execution $\pi$ of Algorithm~\ref{alg:ic}} if
$\pi$ contains the invocation of \textit{acquire($Q$)}, $X_j\in Q$ by
$p_i$ that returned \emph{true}, but does not contain a subsequent
invocation of \textit{release($Q'$)}, $X_j\in Q'$, by $p_i$ in $\pi$. 
\begin{lemma}
\label{lm:mutex}
For any object $X_j$, and any execution
$\pi$ of Algorithm~\ref{alg:ic}, there exists at most one process that \emph{holds} a lock on
$X_j$ after $\pi$. 
\end{lemma}
\begin{proof}
Assume, by contradiction, that there exists an execution $\pi$ after which processes $p_i$ and $p_k$
\emph{hold} a lock on the same object, say $X_j$. 
In order to hold the lock on $X_j$, process $p_i$ writes $1$ to register
$r_{ij}$ and then checks if any other process $p_k$ has written $1$ to $r_{kj}$.  
Since the corresponding operation {\it acquire(Q)}, $X_j \in Q$
invoked by $p_i$ returns {\it true}, $p_i$ read $0$ in $r_{kj}$ in Line~\ref{line:lock}. 
But then $p_k$ also writes $1$ to $r_{kj}$ and later reads that
$r_{ij}$ is 1. 
This is because $p_k$ can write $1$ to $r_{kj}$ only after the read of
$r_{kj}$ returned $0$ to $p_i$ which is preceded by the write of $1$ to
$r_{ij}$. 
Hence, there exists an object $X_j$ such that $r_{ij}=1;i\neq k$, 
but the conditional in Line~\ref{line:lock} returns {\it true} to process $p_k$--- a contradiction. 
\end{proof}
\begin{observation}
\label{ob:write}
Let $\pi$ be any execution of Algorithm~\ref{alg:ic}. Then, for any updating transaction $T_k \in \ms{txns}(\pi)$
executed by process $p_i$
writes to $L_j$ (in Line~\ref{line:wlockwrite}) or $v_j$ (in Line~\ref{line:write}) 
for some $X_j \in \Wset(T_k)$ immediately after $\pi$ \emph{iff} $p_i$ holds the lock on $X_j$ after $\pi$.
\end{observation}
\begin{lemma}
\label{lm:icopaque}
Algorithm~\ref{alg:ic} implements an opaque TM.
\end{lemma}
\begin{proof}
Let $E$ by any finite execution of Algorithm~\ref{alg:ic}. 
Let $<_E$ denote a total-order on events in $E$.

Let $H$ denote a subsequence of $E$ constructed by selecting
\emph{linearization points} of t-operations performed in $E$.
The linearization point of a t-operation $op$, denoted as $\ell_{op}$ is associated with  
a base object event or an event performed between the invocation and response 
of $op$ using the following procedure. 

\vspace{1mm}\noindent\textbf{Completions.}
First, we obtain a completion of $E$ by removing some pending
invocations and adding responses to the remaining pending invocations
involving a transaction $T_k$ as follows:
every incomplete $\Read_k$, $\Write_k$ operation is removed from $E$;
an incomplete $\TryC_k$ is removed from $E$ if $T_k$ has not performed any write to a base object during the \emph{release}
function in Line~\ref{line:rellock}, otherwise it is completed by including $C_k$ after
$E$.

\vspace{1mm}\noindent\textbf{Linearization points.}
Now a linearization $H$ of $E$ is obtained by associating linearization points to
t-operations in the obtained completion of $E$ as follows:
\begin{itemize}
\item For every t-read $op_k$ that returns a non-A$_k$ value, $\ell_{op_k}$ is chosen as the event in Line~\ref{line:read2}
of Algorithm~\ref{alg:ic}, else, $\ell_{op_k}$ is chosen as invocation event of $op_k$
\item For every $op_k=\Write_k $ that returns, $\ell_{op_k}$ is chosen as the invocation event of $op_k$
\item For every $op_k=\TryC_k$ that returns $C_k$ such that $\Wset(T_k)
  \neq \emptyset$, $\ell_{op_k}$ is associated with the response
  of \emph{acquire} in Line~\ref{line:acq}, 
  else if $op_k$ returns $A_k$, $\ell_{op_k}$ is associated with the invocation event of $op_k$
\item For every $op_k=\TryC_k$ that returns $C_k$ such that $\Wset(T_k) = \emptyset$, 
$\ell_{op_k}$ is associated with Line~\ref{line:return}
\end{itemize}
$<_H$ denotes a total-order on t-operations in the complete sequential history $H$.

\vspace{1mm}\noindent\textbf{Serialization points.}
The serialization of a transaction $T_j$, denoted as $\delta_{T_j}$ is
associated with the linearization point of a t-operation 
performed within the execution of $T_j$.

We obtain a t-complete history ${\bar H}$ from $H$ as follows: 
for every transaction $T_k$ in $H$ that is complete, but not t-complete, 
we insert $\textit{tryC}_k\cdot A_k$ after $H$. 
 
A t-complete t-sequential history $S$ is obtained by associating serialization points to transactions in ${\bar H}$ as follows:
\begin{itemize}
\item If $T_k$ is an update transaction that commits, then $\delta_{T_k}$ is $\ell_{\TryC_k}$
\item If $T_k$ is a read-only or aborted transaction in $\bar H$,
$\delta_{T_k}$ is assigned to the linearization point of the last t-read that returned a non-A$_k$ value in $T_k$
\end{itemize}
$<_S$ denotes a total-order on transactions in the t-sequential history $S$.
\begin{claim}
\label{cl:seq}
If $T_i \prec_{H}T_j$, then $T_i <_S T_j$
\end{claim}
\begin{proof}
This follows from the fact that for a given transaction, its
serialization point is chosen between the first and last event of the transaction
implying if $T_i \prec_{H} T_j$, then $\delta_{T_i} <_{E} \delta_{T_j}$ implies $T_i <_S T_j$.
\end{proof}
\begin{claim}
\label{cl:ic1}
Let $T_k$ be any updating transaction that returns \emph{false} from the invocation of \emph{isAbortable}
in Line~\ref{line:abort3}. Then, $T_k$ returns $C_k$ within a finite number of its own steps in any extension of $E$.
\end{claim}
\begin{proof}
Observer that $T_k$ performs the write to base objects $v_j$ for every $X_j \in \Wset(T_k)$ and then invokes
\emph{release} in Lines~\ref{line:write} and \ref{line:rellock} respectively.
Since neither of these involve aborting the transaction or contain unbounded loops or waiting statements, it follows
that $T_k$ will return $C_k$ within a finite number of its steps.
\end{proof}
\begin{claim}
\label{cl:readfrom}
$S$ is legal.
\end{claim}
\begin{proof}
Observe that for every $\Read_j(X_m) \rightarrow v$, there exists some transaction $T_i$
that performs $\Write_i(X_m,v)$ and completes the event in Line~\ref{line:write} such that
$\Read_j(X_m) \not\prec_H^{RT} \Write_i(X_m,v)$.
More specifically, $\Read_j(X_m)$ returns as a non-abort response, the value of the base object $v_m$
and $v_m$ can be updated only by a transaction $T_i$ such that $X_m \in \Wset(T_i)$.
Since $\Read_j(X_m)$ returns the response $v$, the event in Line~\ref{line:read2}
succeeds the event in Line~\ref{line:write} performed by $\TryC_i$.
Consequently, by Claim~\ref{cl:ic1} and the assignment of linearization points,
$\ell_{\TryC_i} <_E \ell_{\Read_j(X_m)}$.
Since, for any updating
committing transaction $T_i$, $\delta_{T_i}=\ell_{\TryC_i}$, by the assignment of serialization points, it follows that
$\delta_{T_{i}} <_E \delta_{T_{j}}$.

Thus, to prove that $S$ is legal, it suffices to show that  
there does not exist a
transaction $T_k$ that returns $C_k$ in $S$ and performs $\Write_k(X_m,v')$; $v'\neq v$ such that $T_i <_S T_k <_S T_j$. 
Suppose that there exists a committed transaction $T_k$, $X_m \in \Wset(T_k)$ such that $T_i <_S T_k <_S T_j$.

$T_i$ and $T_k$ are both updating transactions that commit. Thus, 
\begin{center}
($T_i <_S T_k$) $\Longleftrightarrow$ ($\delta_{T_i} <_{E} \delta_{T_k}$) \\
($\delta_{T_i} <_{E} \delta_{T_k}$) $\Longleftrightarrow$ ($\ell_{\TryC_i} <_{E} \ell_{\TryC_k}$) 
\end{center}
Since, $T_j$ reads the value of $X$ written by $T_i$, one of the following is true:
$\ell_{\TryC_i} <_{E} \ell_{\TryC_k} <_{E} \ell_{\Read_j(X_m)}$ or
$\ell_{\TryC_i} <_{E} \ell_{\Read_j(X_m)} <_{E} \ell_{\TryC_k}$.
Let $T_i$ and $T_k$ be executed by processes $p_i$ and $p_k$ respectively.

Consider the case that $\ell_{\TryC_i} <_{E} \ell_{\TryC_k} <_{E} \ell_{\Read_j(X_m)}$.

By the assignment of linearization points, $T_k$ returns a response from the event in Line~\ref{line:acq} 
before the read of $v_m$ by $T_j$ in Line~\ref{line:read2}. 
Since $T_i$ and $T_k$ are both committed in $E$, $p_k$ returns \emph{true} from the event in
Line~\ref{line:acq} only after $T_i$ writes $0$ to $r_{im}$ in Line~\ref{line:rel1} (Lemma~\ref{lm:mutex}).

Recall that $\Read_j(X_m)$ checks if $X_m$ is locked by a concurrent transaction (i.e $L_j\neq 0$), 
then performs read-validation (Line~\ref{line:abort0}) before returning a matching response. 
Consider the following possible sequence of events: 
$T_k$ returns \emph{true} from the \emph{acquire} function invocation, 
sets $L_j$ to $1$ for every $X_j \in \Wset(T_k)$ (Line~\ref{line:wlockwrite}) and
updates the value of $X_m$ to shared-memory (Line~\ref{line:write}).
The implementation of $\Read_j(X_m)$ then reads the base object $v_m$ associated with $X_m$ after which
$T_k$ releases $X_m$ by writing $0$ to $r_{km}$ and finally $T_j$ performs the check in Line~\ref{line:abort0}. 
However, $\Read_j(X_m)$ is forced to return $A_j$ because $X_m \in \Rset(T_j)$ (Line~\ref{line:rset}) 
and has been invalidated since last reading its value. 
Otherwise suppose that $T_k$ acquires exclusive access to $X_m$ by writing $1$ to $r_{km}$ and returns \emph{true}
from the invocation of \emph{acquire}, updates $v_m$ in Line~\ref{line:write}), 
$T_j$ reads $v_m$, $T_j$ performs the check in Line~\ref{line:abort0} and finally $T_k$ 
releases $X_m$ by writing $0$ to $r_{km}$. 
Again, $\Read_j(X_m)$ returns $A_j$ since $T_j$ reads that $r_{km}$ is $1$---contradiction.

Thus, $\ell_{\TryC_i} <_E \ell_{\Read_j(X)} <_{E} \ell_{\TryC_k}$.

We now need to prove that $\delta_{T_{j}}$ indeed precedes $\ell_{\TryC_k}$ in $E$.

Consider the two possible cases:
\begin{itemize}
\item
Suppose that $T_j$ is a read-only or aborted transaction in $\bar H$. 
Then, $\delta_{T_j}$ is assigned to the last t-read performed by $T_j$ that returns a non-A$_j$ value. 
If $\Read_j(X_m)$ is not the last t-read performed by $T_j$ that returned a non-A$_j$ value, 
then there exists a $\Read_j(X_z)$ performed by $T_j$ such that 
$\ell_{\Read_j(X_m)} <_{E} \ell_{\TryC_k} <_E \ell_{\Read_j(X_z)}$.
Now assume that $\ell_{\TryC_k}$ must precede $\ell_{\Read_j(X_z)}$ to obtain a legal $S$.
Since $T_k$ and $T_j$ are concurrent in $E$, we are restricted to the case that
$T_k$ performs a $\Write_k(X_z,v)$ and $\Read_j(X_z)$ returns $v$.
However, we claim that this t-read of $X_z$ must abort by performing the checks in Line~\ref{line:abort0}.
Observe that $T_k$ writes $1$ to $L_m$, $L_z$ each (Line~\ref{line:wlockwrite}) and 
then writes new values to base objects $v_m$, $v_z$ (Line~\ref{line:write}).
Since $\Read_j(X_z)$ returns a non-$A_j$ response, $T_k$ writes $0$ to $L_z$ before the read
of $L_z$ by $\Read_j(X_z)$ in Line~\ref{line:abort0}.
Thus, the t-read of $X_z$ would return $A_j$ (in Line~\ref{line:read-validate} after validation of the read set since $X_m$
has been updated---
contradiction to the assumption that it the last t-read by $T_j$ to return a non-$A_j$ response.
\item
Suppose that $T_j$ is an updating transaction that commits, then $\delta_{T_j}=\ell_{\TryC_j}$ which implies that
$\ell_{\Read_j(X_m)} <_{E} \ell_{\TryC_k} <_E \ell_{\TryC_j}$. Then, $T_j$ must necessarily perform the checks
in Line~\ref{line:abort3} and read that $L_m$ is $1$. 
Thus, $T_j$ must return $A_j$---contradiction to the assumption that $T_j$ is a committed transaction.
\end{itemize}
\end{proof}
The conjunction of Claims~\ref{cl:seq} and \ref{cl:readfrom} establish that Algorithm~\ref{alg:ic} is opaque.
\end{proof}
\setcounter{theorem}{8}
\begin{theorem}
\label{th:ic}
Algorithm~\ref{alg:ic} describes a progressive, opaque and strict DAP TM implementation $LP$ that provides
wait-free TM-liveness, uses invisible reads and in every execution $E$ of $LP$, 
\begin{itemize}
\item every transaction $T\in \ms{txns}(E)$ applies only read and write primitives in $E$,
\item every transaction $T \in \ms{txns}(E)$ performs at most a single RAW,
\item for every transaction $T \in \ms{txns}(E)$, every t-read operation performed by $T$
incurs $O(1)$ memory stalls in $E$.
\end{itemize}
\end{theorem}
\begin{proof}
\textit{(TM-liveness and TM-progress)}
Since none of the implementations of the t-operations in Algorithm~\ref{alg:ic}
contain unbounded loops or waiting statements, every t-operation $op_k$ returns a matching response
after taking a finite number of steps in every execution. Thus, Algorithm~\ref{alg:ic}
provides wait-free TM-liveness.

To prove progressiveness, we proceed by enumerating the cases under which a transaction $T_k$ may be aborted.
\begin{itemize}
\item
Suppose that there exists a $\Read_k(X_j)$ performed by $T_k$ that returns $A_k$
from Line~\ref{line:abort0}.
Thus, there exists a process $p_t$ executing a transaction
that has written $1$ to $r_{tj}$ in Line~\ref{line:acq1}, but has not yet written
$0$ to $r_{tj}$ in Line~\ref{line:rel1} or
some t-object in $\Rset(T_k)$ has been updated since its t-read by $T_k$.
In both cases, there exists a concurrent transaction performing a 
t-write to some t-object in $\Rset(T_k)$.
\item
Suppose that $\TryC_k$ performed by $T_k$ that returns $A_k$
from Line~\ref{line:abort2}.
Thus, there exists a process $p_t$ executing a transaction
that has written $1$ to $r_{tj}$ in Line~\ref{line:acq1}, but has not yet written
$0$ to $r_{tj}$ in Line~\ref{line:rel1}. Thus, $T_k$ encounters step-contention with another
transaction that concurrently attempts to update a t-object in $\Wset(T_k)$.
\item
Suppose that $\TryC_k$ performed by $T_k$ that returns $A_k$
from Line~\ref{line:abort3}.
Since $T_k$ returns $A_k$ from Line~\ref{line:abort3} for the same reason it
returns $A_k$ after Line~\ref{line:abort0}, the proof follows.
\end{itemize}
\textit{(Strict disjoint-access parallelism)}
Consider any execution $E$ of Algorithm~\ref{alg:ic} and let $T_i$
and $T_j$ be any two transactions that participate in $E$ and access the same
base object $b$ in $E$.
\begin{itemize}
\item
Suppose that $T_i$ and $T_j$ contend on base object $v_j$ or $L_j$.
Since for every t-object $X_j$, there exists distinct base objects $v_j$ and $L_j$,
$T_j$ and $T_j$ contend on $v_j$ only if $X_j \in \Dset(T_i) \cap \Dset(T_j)$.
\item
Suppose that $T_i$ and $T_j$ contend on base object $r_{ij}$.
Without loss of generality, let $p_i$ be the process executing 
transaction $T_i$; $X_j \in \Wset(T_i)$ that writes $1$ to $r_{ij}$ in Line~\ref{line:acq1}.
Indeed, no other process executing a transaction that writes to $X_j$ can write to $r_{ij}$.
Transaction $T_j$ reads $r_{ij}$ only if $X_j \in \Dset(T_j)$ as evident from the accesses performed
in Lines~\ref{line:acq1}, \ref{line:lock}, \ref{line:rel1}, \ref{line:isl}.
\end{itemize}
Thus, $T_i$ and $T_j$ access the same base object only if they access a common t-object.

\textit{(Opacity)}
Follows from Lemma~\ref{lm:icopaque}.

\textit{(Invisible reads)}
Observe that read-only transactions do not perform any nontrivial events.
Secondly, in any execution $E$ of Algorithm~\ref{alg:ic}, and any transaction $T_k\in \ms{txns}(E)$,
if $X_j\in \Rset(T_k)$, $T_k$ does not write to any of the base objects associated with $X_j$ nor
write any information that reveals its read set to other transactions.

\textit{(Complexity)} 
Consider any execution $E$ of Algorithm~\ref{alg:ic}.
\begin{itemize}
\item
For any $T_k \in \ms{txns}(E)$, each $\Read_k$ only applies trivial primitives in $E$ while $\TryC_k$ simply
returns $C_k$ if $\Wset(T_k)=\emptyset$. Thus, Algorithm~\ref{alg:ic} uses invisible reads.
\item
Any read-only transaction $T_k \in \ms{txns}(E)$ not perform any RAW or AWAR.
An updating transaction $T_k$ executed by process $p_i$ performs a sequence of writes (Line~\ref{line:acq1}
to base objects $\{r_{ij}\}:X_j \in \Wset(T_k)$, followed by a sequence of reads to base objects 
$\{r_{tj}\}:t\in \{1,\ldots , n\}, X_j \in \Wset(T_k)$
(Line~\ref{line:lock}) thus incurring a single multi-RAW.
\item
Let $e$ be a write event performed by some transaction $T_k$ executed by process $p_i$ in $E$ on 
base objects $v_j$ and $L_j$ (Lines~\ref{line:write} and \ref{line:wlockwrite}).
Any transaction $T_k$ performs a write to $v_j$ or $L_j$ only after $T_k$ writes $0$ to $r_{ij}$, for every $X_j\in \Wset(T_k)$.
Thus, by Lemmata~\ref{lm:mutex} and \ref{lm:icopaque}, it follows that
events that involve an access to either of these base objects incurs $O(1)$ stalls.

Let $e$ be a write event on base object $r_{ij}$ (Line~\ref{line:acq1}) while writing to t-object $X_j$.
By Algorithm~\ref{alg:ic}, no other process can write to $r_{ij}$.
It follows that any transaction $T_k \in \ms{txns}(E)$ incurs $O(1)$ memory stalls 
on account of any event it performs in $E$.
%
Observe that any t-read $\Read_k(X_j)$ only accesses base objects $v_j$, $L_j$ and other value base objects in $\Rset(T_k)$.
But as already established above, these are $O(1)$ stall events. Hence, every t-read operation
incurs $O(1)$-stalls in $E$.
\end{itemize}
\end{proof}
\section{Obstruction-free TMs}
\label{app:upperof}
\subsection{An opaque RW DAP TM implementation $M\in \mathcal{OF}$}
\label{app:rwoftm}
\begin{algorithm}[!h]
\caption{RW DAP opaque implementation $M\in \mathcal{OF}$; code for $T_k$
}\label{alg:oftm}
  \begin{algorithmic}[1]
  	\begin{multicols}{2}
  	{\footnotesize
	\Part{Shared base objects}{
		\State $\ms{tvar}[m]$, storing $[\ms{owner}_m,\ms{oval}_m,\ms{nval}_m]$
		\State ~~~~for each t-object $X_m$, supports read, write, cas
		\State ~~~~$\ms{owner}_m$, a transaction identifier 
		\State ~~~~$\ms{oval}_m\in V$
		\State ~~~~$\ms{nval}_m\in V$
		\State $\ms{status}[k] \in \{\ms{live},\ms{aborted},\ms{committed}\}$,  
		\State ~~~~for each $T_k$; supports read, write, cas
	}\EndPart	
	\Part{Local variables}{
		\State $\ms{Rset}_k,\ms{Wset}_k$ for every transaction $T_k$;
		\State ~~~~dictionaries storing $\{X_m$, $\ms{Tvar}[m]\}$
	}\EndPart	
	\Statex
	\Part{\Read$_k(X_m)$}{
		\State $[\ms{owner}_m,\ms{oval}_m,\ms{nval}_m]$ $\gets$ $\ms{tvar}[m].\lit{read}()$ \label{line:linr}
		
		\If{$\ms{owner}_m \neq k$}
			
			\State $s_m\gets \ms{status}[\ms{owner}_m].\lit{read}()$ \label{line:status1}
			\If{$s_m=\ms{committed}$} \label{line:rcurr}
				\State $\ms{curr}=\ms{nval}_m$
			\ElsIf{$s_m=\ms{aborted}$} 
			  \State $\ms{curr}=\ms{oval}_m$
			  
			\Else
				\If{$\ms{status}[\ms{owner}_m].\lit{cas}(\ms{live},\ms{aborted})$} \label{line:awar}
				  \State $\ms{curr}=\ms{oval}_m$
				\Else
				\Return $A_k$ \EndReturn
				\EndIf
				
			\EndIf
			\If{$ \ms{status}[k]=\ms{live} \wedge \neg \lit{validate}()$} \label{line:rc1}
				\State $\Rset(T_k).\lit{add}(\{X_m,[\ms{owner}_m,\ms{oval}_m,\ms{nval}_m]\})$
				\Return $\ms{curr}$ \EndReturn
			\EndIf
			\Return $A_k$ \EndReturn \label{line:of2}
			
		\Else
			
			\Return $\Rset(T_k).\lit{locate}(X_m)$ \EndReturn
				
		\EndIf

   	 }\EndPart
	\Statex
	\Part{Function: $\lit{validate}()$}{
		\If{$\exists \{X_j,[\ms{owner}_j,\ms{oval}_j,\ms{nval}_j]\} \in \Rset(T_k)$:\\
		~~~~~~~($[\ms{owner}_j,\ms{oval}_j,\ms{nval}_j]\neq \ms{tvar}[j].\lit{read}())$}
			\Return $\true$ \EndReturn
		\EndIf
		\Return $\false$ \EndReturn
	}\EndPart
	
	\newpage
	\Part{\Write$_k(X_m,v)$}{
		\State $[\ms{owner}_m,\ms{oval}_m,\ms{nval}_m]$ $\gets$ $\ms{tvar}[m].\lit{read}()$ \label{line:writeread}
		\If{$\ms{owner}_m \neq k$}
			
			\State $s_m\gets \ms{status}[\ms{owner}_m].\lit{read}()$ \label{line:status2}
			\If{$s_m=\ms{committed}$} \label{line:nval}
				\State $\ms{curr}=\ms{nval}_m$
			\ElsIf{$s_m=\ms{aborted}$} 
			  \State $\ms{curr}=\ms{oval}_m$
			\Else
				\If{$\ms{status}[\ms{owner}_m].\lit{cas}(\ms{live},\ms{aborted})$} \label{line:writeabort}
				  \State $\ms{curr}=\ms{oval}_m$
				
				\Else
				\Return $A_k$ \EndReturn
			\EndIf
				
		\EndIf
		\State $o_m\gets \ms{tvar}[m].\lit{cas}([\ms{owner}_m,\ms{oval}_m,\ms{nval}_m],[k,\ms{curr},v])$ \label{line:linw}
		\If{$o_m\wedge \ms{status}[k]= \ms{live}$} \label{line:wc1}
			\State $\ms{Wset}_k.\lit{add}(\{X_m,[k,\ms{curr},v]\})$
				\Return $ok$ \EndReturn
		\Else
			\Return $A_k$ \EndReturn \label{line:of4}
		\EndIf
		\Else
				\State $[\ms{owner}_m,\ms{oval}_m,\ms{nval}_m]=\ms{Wset}_k.\lit{locate}(X_m)$
				\State $s= \ms{tvar}[m].\lit{cas}([\ms{owner}_m,\ms{oval}_m,\ms{nval}_m],[k,\ms{oval}_m,v])$
				\If{$s$}
					\State $\Wset(T_k).\lit{add}(\{X_m,[k,\ms{oval}_m,v]\})$
					\Return $ok$ \EndReturn
				\Else
					\Return $A_k$ \EndReturn \label{line:of5}
		
				\EndIf
		\EndIf
   	}\EndPart	
   	
   	\Statex
	\Part{\TryC$_k$()}{
		\If{$\lit{validate}()$}
			\Return $A_k$ \EndReturn  \label{line:of1}
		\EndIf
		\If{$\ms{status}[k].\lit{cas}(\ms{live},\ms{committed})$} \label{line:tryc}
			\Return $C_k$ \EndReturn
		\EndIf
		\Return $A_k$ \EndReturn 
   	 }\EndPart

	}\end{multicols}
  \end{algorithmic}
\end{algorithm}
\begin{lemma}
\label{lm:oftmopaque}
Algorithm~\ref{alg:oftm} implements an opaque TM.
\end{lemma}
\begin{proof}
Since opacity is a safety property, we only consider finite executions~\cite{icdcs-opacity}.
Let $E$ by any finite execution of Algorithm~\ref{alg:oftm}. 
Let $<_E$ denote a total-order on events in $E$.

Let $H$ denote a subsequence of $E$ constructed by selecting
\emph{linearization points} of t-operations performed in $E$.
The linearization point of a t-operation $op$, denoted as $\ell_{op}$ is associated with  
a base object event or an event performed during 
the execution of $op$ using the following procedure. 

\vspace{1mm}\noindent\textbf{Completions.}
First, we obtain a completion of $E$ by removing some pending
invocations and adding responses to the remaining pending invocations
involving a transaction $T_k$ as follows:
every incomplete $\Read_k$, $\Write_k$, $\TryC_k$ operation is removed from $E$;
an incomplete $\Write_k$ is removed from $E$.

\vspace{1mm}\noindent\textbf{Linearization points.}
We now associate linearization points to
t-operations in the obtained completion of $E$ as follows:
\begin{itemize}
\item For every t-read $op_k$ that returns a non-A$_k$ value, $\ell_{op_k}$ is chosen as the event in Line~\ref{line:linr}
of Algorithm~\ref{alg:oftm}, else, $\ell_{op_k}$ is chosen as invocation event of $op_k$
\item For every t-write $op_k$ that returns a non-A$_k$ value, $\ell_{op_k}$ is chosen as the event in Line~\ref{line:writeread}
of Algorithm~\ref{alg:oftm}, else, $\ell_{op_k}$ is chosen as invocation event of $op_k$
\item For every $op_k=\TryC_k$ that returns $C_k$, $\ell_{op_k}$ is associated with Line~\ref{line:of1}.
\end{itemize}
$<_H$ denotes a total-order on t-operations in the complete sequential history $H$.

\vspace{1mm}\noindent\textbf{Serialization points.}
The serialization of a transaction $T_j$, denoted as $\delta_{T_j}$ is
associated with the linearization point of a t-operation 
performed during the execution of the transaction.

We obtain a t-complete history ${\bar H}$ from $H$ as follows: 
for every transaction $T_k$ in $H$ that is complete, but not t-complete, 
we insert $\textit{tryC}_k\cdot A_k$ after $H$. 

${\bar H}$ is thus a t-complete sequential history.
A t-complete t-sequential history $S$ equivalent to ${\bar H}$ is obtained by associating 
serialization points to transactions in ${\bar H}$ as follows:
\begin{itemize}
\item If $T_k$ is an update transaction that commits, then $\delta_{T_k}$ is $\ell_{tryC_k}$
\item If $T_k$ is an aborted or read-only transaction in $\bar H$,
then $\delta_{T_k}$ is assigned to the linearization point of the last t-read that returned a non-A$_k$ value in $T_k$
\end{itemize}
$<_S$ denotes a total-order on transactions in the t-sequential history $S$.
\begin{claim}
\label{cl:oseq}
If $T_i \prec_{H}^{RT} T_j$, then $T_i <_S T_j$.
\end{claim}
\begin{proof}
This follows from the fact that for a given transaction, its
serialization point is chosen between the first and last event of the transaction
implying if $T_i \prec_{H} T_j$, then $\delta_{T_i} <_{E} \delta_{T_j}$ implies $T_i <_S T_j$ 
\end{proof}
\begin{claim}
\label{cl:ofclaim0}
If transaction $T_i$ returns $C_i$ in $E$, then \emph{status[i]=committed} in $E$.
\end{claim}
\begin{proof}
Transaction $T_i$ must perform the event in Line~\ref{line:tryc}
before returning $T_i$ i.e. the \emph{cas} on its own \emph{status} to change the
value to \emph{committed}. The proof now follows from the fact that any other transaction
may change the \emph{status} of $T_i$ only if it is \emph{live} (Lines~\ref{line:writeabort} and \ref{line:awar}).
\end{proof}
\begin{claim}
\label{cl:oreadfrom}
$S$ is legal.
\end{claim}
\begin{proof}
Observe that for every $\Read_j(X) \rightarrow v$, there exists some transaction $T_i$
that performs $\Write_i(X,v)$ and completes the event in Line~\ref{line:linw} to write $v$ as the \emph{new value} of $X$ such that
$\Read_j(X) \not\prec_H^{RT} \Write_i(X,v)$. 
For any updating
committing transaction $T_i$, $\delta_{T_i}=\ell_{\TryC_i}$. 
Since $\Read_j(X)$ returns a response $v$, the event in Line~\ref{line:linr} must succeed
the event in Line~\ref{line:tryc} when $T_i$ changes \emph{status[i]} to \emph{committed}.
Suppose otherwise, then $\Read_j(X)$ subsequently forces $T_i$ to abort by writing
\emph{aborted} to \emph{status[i]} and must return the \emph{old value} of $X$ 
that is updated by the previous \emph{owner} of $X$, which must be committed in $E$ (Line~\ref{line:nval}).
Since $\delta_{T_{i}}=\ell_{\TryC_{i}}$ precedes the event in Line~\ref{line:tryc},
it follows that $\delta_{T_{i}} <_E \ell_{\Read_{j}(X)}$.

We now need to prove that $\delta_{T_{i}} <_E \delta_{T_{j}}$. Consider the following cases:
\begin{itemize}
\item
if $T_j$ is an updating committed transaction, then $\delta_{T_{j}}$ is assigned to $\ell_{\TryC_{j}}$.
But since $\ell_{\Read_{j}(X)} <_E \ell_{\TryC_{j}}$, it follows that $ \delta_{T_{i}} <_E \delta_{T_{j}}$.
\item
if $T_j$ is a read-only or aborted transaction, then $\delta_{T_{j}}$ is assigned to
the last t-read that did not abort. Again, it follows that $ \delta_{T_{i}} <_E \delta_{T_{j}}$.
\end{itemize}
To prove that $S$ is legal, we need to show that,
there does not exist any
transaction $T_k$ that returns $C_k$ in $S$ and performs $\Write_k(X,v')$; $v'\neq v$ such that $T_i <_S T_k <_S T_j$. 
Now, suppose by contradiction that there exists a committed transaction $T_k$, $X \in \Wset(T_k)$ that writes $v'\neq v$ to $X$ 
such that $T_i <_S T_k <_S T_j$.
Since $T_i$ and $T_k$ are both updating transactions that commit,
\begin{center}
($T_i <_S T_k$) $\Longleftrightarrow$ ($\delta_{T_i} <_{E} \delta_{T_k}$) \\
($\delta_{T_i} <_{E} \delta_{T_k}$) $\Longleftrightarrow$ ($\ell_{\TryC_i} <_{E} \ell_{\TryC_k}$) 
\end{center}
Since, $T_j$ reads the value of $X$ written by $T_i$, one of the following is true:
$\ell_{\TryC_i} <_{E} \ell_{\TryC_k} <_{E} \ell_{\Read_j(X)}$ or
$\ell_{\TryC_i} <_{E} \ell_{\Read_j(X)} <_{E} \ell_{\TryC_k}$.

If $\ell_{\TryC_i} <_{E} \ell_{\TryC_k} <_{E} \ell_{\Read_j(X)}$, then the event in Line~\ref{line:tryc}
performed by $T_k$ when it changes the status field to \emph{committed}
precedes the event in Line~\ref{line:linr} performed by $T_j$.
Since $\ell_{\TryC_i} <_{E} \ell_{\TryC_k}$ and both $T_i$ and $T_k$ are committed in $E$,
$T_k$ must perform the event in Line~\ref{line:writeread} after $T_i$ changes \emph{status[i]}
to \emph{committed} since otherwise, $T_k$ would perform the event in Line~\ref{line:writeabort}
and change \emph{status[i]} to \emph{aborted}, thereby forcing $T_i$ to return $A_i$.
However, $\Read_j(X)$ observes that the \emph{owner} of $X$ is $T_k$
and since the \emph{status} of $T_k$ is committed at this point in the execution,
$\Read_j(X)$ must return $v'$ and not $v$---contradiction.

Thus, $\ell_{\TryC_i} <_{E} \ell_{\Read_j(X)} <_{E} \ell_{\TryC_k}$.
We now need to prove that $\delta_{T_{j}}$ indeed precedes $\delta_{T_{k}}=\ell_{\TryC_k}$ in $E$.

Now consider two cases:
\begin{itemize}
\item
Suppose that $T_j$ is a read-only transaction. 
Then, $\delta_{T_j}$ is assigned to the last t-read performed by $T_j$ that returns a non-A$_j$ value. 
If $\Read_j(X)$ is not the last t-read that returned a non-A$_j$ value, then there exists a $read_j(X')$ such that 
$\ell_{\Read_j(X)} <_{E} \ell_{\TryC_k} <_E \ell_{read_j(X')}$.
But then this t-read of $X'$ must abort since the value of $X$ has been updated by $T_k$ since $T_j$ first 
read $X$---contradiction.
\item
Suppose that $T_j$ is an updating transaction that commits, then $\delta_{T_j}=\ell_{\TryC_j}$ which implies that
$\ell_{read_j(X)} <_{E} \ell_{\TryC_k} <_E \ell_{\TryC_j}$. Then, $T_j$ must neccesarily perform the validation
of its read set in Line~\ref{line:of1} and return $A_j$---contradiction.
\end{itemize}
\end{proof}
Claims~\ref{cl:oseq} and \ref{cl:oreadfrom} establish that
Algorithm~\ref{alg:oftm} is opaque.
\end{proof}

\begin{theorem}
\label{th:ofraw}
Algorithm~\ref{alg:oftm} describes a RW DAP, opaque TM implementation $M\in \mathcal{OF}$ such that every
execution $E$ of $M$ is a $O(n)$-stall execution for any t-read operation and
every read-only transaction $T \in \ms{txns}(E)$ performs $O(|\Rset(T)|)$ AWARs in $E$.
\end{theorem}
\begin{proof}
\textit{(Opacity)}
Follows from Lemma~\ref{lm:oftmopaque}

\textit{(TM-liveness and TM-progress)}
Since none of the implementations of the t-operations in Algorithm~\ref{alg:oftm}
contain unbounded loops or waiting statements, every
t-operation $op_k$ returns a matching response after taking a finite number of steps.
Thus, Algorithm~\ref{alg:oftm} provides wait-free TM-liveness.

To prove OF TM-progress, we proceed by enumerating the cases under which a transaction $T_k$ may be aborted in any execution.
\begin{itemize}
\item
Suppose that there exists a $\Read_k(X_m)$ performed by $T_k$ that returns $A_k$.
If $\Read_k(X_m)$ returns $A_k$ in Line~\ref{line:of2}, then there exists a concurrent transaction
that updated a t-object in $\Rset(T_k)$ or changed \emph{status[k]} to \emph{aborted}.
In both cases, $T_k$ returns $A_k$ only because there is step contention.
\item
Suppose that there exists a $\Write_k(X_m,v)$ performed by $T_k$ that returns $A_k$ in Line~\ref{line:of4}.
Thus, either a concurrent transaction has changed \emph{status[k]} to \emph{aborted} or the value
in $\ms{tvar}[m]$ has been updated since the event in Line~\ref{line:writeread}.
In both cases, $T_k$ returns $A_k$ only because of step contention with another transaction.
\item
Suppose that a $\Read_k(X_m)$ or $\Write_k(X_m,v)$ return $A_k$ in Lines~\ref{line:awar} and \ref{line:writeabort} respectively.
Thus, a concurrent transaction has takes steps concurrently by updating the 
\emph{status} of $\ms{owner}_m$ since the read by $T_k$ in Lines~\ref{line:linr} and \ref{line:writeread} respectively.
\item
Suppose that $\TryC_k()$ returns $A_k$ in Line~\ref{line:of5}. This is because
there exists a t-object in $\Rset(T_k)$ that has been updated by a concurrent transaction since
i.e. $\TryC_k()$ returns $A_k$ only on encountering step contention.
\end{itemize}
It follows that in any step contention-free execution of a transaction $T_k$
from a $T_k$-free execution, $T_k$ must return $C_k$ after taking a finite number of steps.

\textit{(Read-write disjoint-access parallelism)}
Consider any execution $E$ of Algorithm~\ref{alg:oftm} and let
$T_i$ and $T_j$ be any two transactions
that contend on a base object $b$ in $E$.
We need to prove that there is a path between a t-object in $\Dset(T_i)$ and a t-object in $\Dset(T_j)$ 
in ${\tilde G}(T_i,T_j,E)$ or there exists $X \in \Dset(T_i) \cap \Dset(T_j)$.
Recall that there exists an edge between t-objects $X$ and $Y$ in ${\tilde G}(T_i,T_j,E)$
only if there exists a transaction $T\in \ms{txns}(E)$ such that $\{X,Y\} \in \Wset(T)$.
\begin{itemize}
\item
Suppose that $T_i$ and $T_j$ contend on base object \ms{tvar}[m] belonging to t-object $X_m$ in $E$.
By Algorithm~\ref{alg:oftm}, a transaction accesses $X_m$ only if $X_m$ is contained in $\Dset(T_m)$.
Thus, both $T_i$ and $T_j$ must access $X_m$.
\item
Suppose that $T_i$ and $T_j$ contend on base object \ms{status}[i] in $E$ (the case when $T_i$ and $T_j$ contend
on \ms{status}[j] is symmetric).
$T_j$ accesses \emph{status[i]} while performing a t-read of some t-object $X$ in Lines~\ref{line:status1} and \ref{line:awar}
only if $T_i$ is the \emph{owner} of $X$.
Also, $T_j$ accesses \emph{status[i]} while performing a t-write to $X$ in Lines~\ref{line:status2} and \ref{line:writeabort}
only if $T_i$ is the \emph{owner} of $X$.
But if $T_i$ is the \emph{owner} of $X$, then $X \in \Wset(T_i)$.
\item
Suppose that $T_i$ and $T_j$ contend on base object \ms{status}[m] belonging to some transaction $T_m$ in $E$.
Firstly, observe that $T_i$ or $T_j$ access \emph{status[m]} only if 
there exist t-objects $X$ and $Y$ in $\Dset(T_i)$ and $\Dset(T_j)$ respectively such that $\{X,Y\} \in \Wset(T_m)$.
This is because $T_i$ and $T_j$ would both read \emph{status[m]} in Lines~\ref{line:status1} (during t-read) and 
\ref{line:status2} (during t-write)
only if $T_m$ was the previous \emph{owner} of $X$ and $Y$.
Secondly, one of $T_i$ or $T_j$ applies a nontrivial primitive to \emph{status[m]}
only if $T_i$ and $T_j$ read \emph{status[m]=live} in Lines~\ref{line:status1} (during t-read) and 
\ref{line:writeread} (during t-write).
Thus, at least one of $T_i$ or $T_j$ is concurrent to $T_m$ in $E$.
It follows that there exists a path between $X$ and $Y$ 
in ${\tilde G}(T_i,T_j,E)$.
\end{itemize}
\textit{(Complexity)} 
Every t-read operation performs at most one AWAR in an execution $E$ (Line~\ref{line:awar}) of Algorithm~\ref{alg:oftm}.
It follows that any read-only transaction $T_k \in \ms{txns}(E)$ performs at most $|\Rset(T_k)|$ AWARs in $E$.

The linear step-complexity is immediate from the fact that during the t-read operations, the transaction validates its entire read
set (Line~\ref{line:rc1}). All other t-operations incur $O(1)$ step-complexity since they involve no iteration statements
like \emph{for} and \emph{while} loops.

Since at most $n-1$ transactions may be t-incomplete at any point in an execution $E$, it follows that
$E$ is at most a $(n-1)$-stall execution for any t-read $op$ and every $T\in \ms{txns}(E)$
incurs $O(n)$ stalls on account of any event performed in $E$. More specifically, consider the following
execution $E$: for all $i\in \{1,\ldots , n-1\}$, each transaction $T_i$ performs $\Write_i(X_m,v)$
in a step-contention free execution until it is poised to apply a nontrivial event on $\ms{tvar}[m]$ (Line~\ref{line:linw}).
By OF TM-progress, we construct $E$ such that each of the $T_i$ is poised to apply a nontrivial event on $\ms{tvar}[m]$
after $E$. Consider the execution fragment of $\Read_n(X_m)$ that is poised to perform an event $e$
that reads $\ms{tvar}[m]$ (Line~\ref{line:linr}) immediately after $E$.
In the constructed execution, $T_n$ incurs $O(n)$ stalls on account of $e$ and thus, produces the desired $(n-1)$-stall execution
for $\Read_n(X)$.
\end{proof}
\subsection{An opaque weak DAP implementation $M \in \mathcal{OF}$}
\label{app:woftm}
\begin{algorithm}[!h]
\caption{Weak DAP opaque implementation $M\in \mathcal{OF}$; code for $T_k$
}\label{alg:oftm2}
  \begin{algorithmic}[1]
  	{\footnotesize
	
	\Part{\Read$_k(X_m)$}{
		\State $[\ms{owner}_m,\ms{oval}_m,\ms{nval}_m]$ $\gets$ $\ms{tvar}[m].\lit{read}()$ \label{line:ofread1}
		
		\If{$\ms{owner}_m \neq k$}
			
			\State $s_m\gets \ms{status}[\ms{owner}_m].\lit{read}()$ \label{line:ownerread}
			\If{$s_m=\ms{committed}$} 
				\State $\ms{curr}=\ms{nval}_m$
			\ElsIf{$s_m=\ms{aborted}$} 
			  \State $\ms{curr}=\ms{oval}_m$

			\Else
				\If{$\ms{status}[\ms{owner}_m].\lit{cas}(\ms{live},\ms{aborted})$} \label{line:ownerwrite}
				  \State $\ms{curr}=\ms{oval}_m$
				\EndIf
				\Return $A_k$ \EndReturn
			\EndIf
				
			\State $o_m\gets \ms{tvar}[m].\lit{cas}([\ms{owner}_m,\ms{oval}_m,\ms{nval}_m],[k,\ms{oval}_m,\ms{nval}_m])$ \label{line:readowner}
			\If{$o_m\wedge \ms{status}[k]= \ms{live}$} 
			    \State $\Rset(T_k).\lit{add}(\{X_m,[\ms{owner}_m,\ms{oval}_m,\ms{nval}_m]\})$
				\Return $\ms{curr}$ \EndReturn
			\EndIf
			
		\Else
			
			\Return $\Rset(T_k).\lit{locate}(X_m)$ \EndReturn
				
		\EndIf

   	 }\EndPart
	\Statex
	\Part{\TryC$_k$()}{
		
		\If{$\ms{status}[k].\lit{cas}(\ms{live},\ms{committed})$} 
			\Return $C_k$ \EndReturn
		\EndIf
		\Return $A_k$ \EndReturn 
   	 }\EndPart

	}
  \end{algorithmic}
\end{algorithm}
Algorithm~\ref{alg:oftm2} describes a weak DAP implementation in $\mathcal{OF}$ that does not satisfy read-write DAP.
The code for the t-write operations is identical to Algorithm~\ref{alg:oftm}.
\begin{theorem}
\label{th:ofweakdap}
Algorithm~\ref{alg:oftm2} describes a weak TM implementation $M\in \mathcal{OF}$ such that
in any execution $E$ of $M$, for every transaction $T \in \ms{txns}(E)$,
$T$ performs $O(1)$ steps during the execution of any t-operation in $E$.
\end{theorem}
\begin{proof}
The proofs of opacity, TM-liveness and TM-progress are almost identical to the analogous proofs
for Algorithm~\ref{alg:oftm}.

\textit{(Weak disjoint-access parallelism)}
Consider any execution $E$ of Algorithm~\ref{alg:oftm2} and let
$T_i$ and $T_j$ be any two transactions
that contend on a base object $b$ in $E$.
We need to prove that there is a path between a t-object in $\Dset(T_i)$ and a t-object in $\Dset(T_j)$ 
in ${\tilde G}(T_i,T_j,E)$ or there exists $X \in \Dset(T_i) \cap \Dset(T_j)$.
Recall that there exists an edge between t-objects $X$ and $Y$ in ${G}(T_i,T_j,E)$
only if there exists a transaction $T\in \ms{txns}(E)$ such that $\{X,Y\} \in \Dset(T)$.
\begin{itemize}
\item
Suppose that $T_i$ and $T_j$ contend on base object \ms{tvar}[m] belonging to t-object $X_m$ in $E$.
By Algorithm~\ref{alg:oftm2}, a transaction accesses $X_m$ only if $X_m$ is contained in $\Dset(T_m)$.
Thus, both $T_i$ and $T_j$ must access $X_m$.
\item
Suppose that $T_i$ and $T_j$ contend on base object \ms{status}[i] in $E$ (the case when $T_i$ and $T_j$ contend
on \ms{status}[j] is symmetric).
$T_j$ accesses \emph{status[i]} while performing a t-read of some t-object $X$ in 
Lines~\ref{line:ownerread} and \ref{line:ownerwrite}
only if $T_i$ is the \emph{owner} of $X$.
Also, $T_j$ accesses \emph{status[i]} while performing a t-write to $X$ in Lines~\ref{line:status2} and \ref{line:writeabort}
only if $T_i$ is the \emph{owner} of $X$.
But if $T_i$ is the \emph{owner} of $X$, then $X \in \Dset(T_i)$.
\item
Suppose that $T_i$ and $T_j$ contend on base object \ms{status}[m] belonging to some transaction $T_m$ in $E$.
Firstly, observe that $T_i$ or $T_j$ access \emph{status[m]} only if 
there exist t-objects $X$ and $Y$ in $\Dset(T_i)$ and $\Dset(T_j)$ respectively such that $\{X,Y\} \in \Dset(T_m)$.
This is because $T_i$ and $T_j$ would both read \emph{status[m]} in Lines~\ref{line:ownerread} (during t-read) and 
\ref{line:status2} (during t-write)
only if $T_m$ was the previous \emph{owner} of $X$ and $Y$.
Secondly, one of $T_i$ or $T_j$ applies a nontrivial primitive to \emph{status[m]}
only if $T_i$ and $T_j$ read \emph{status[m]=live} in Lines~\ref{line:ownerread} (during t-read) and 
\ref{line:writeread} (during t-write).
Thus, at least one of $T_i$ or $T_j$ is concurrent to $T_m$ in $E$.
It follows that there exists a path between $X$ and $Y$ 
in ${\tilde G}(T_i,T_j,E)$.
\end{itemize}
\textit{(Complexity)}
Since no implementation of any of the t-operation contains any iteration statements like \emph{for} and \emph{while} loops), 
the proof follows.
\end{proof}
\end{document}